\documentclass[12pt]{article}
\usepackage{amsthm}
\usepackage{bbm}
\usepackage{graphicx}
\usepackage{subfigure}
\graphicspath{{figs/}}
\usepackage{float}
\usepackage{epsfig}
\usepackage{paralist}
\usepackage{epstopdf}
\usepackage[colorlinks=true, citecolor=blue]{hyperref}
\usepackage{indentfirst}
\usepackage{epsfig,graphicx,amssymb,amsmath}

\usepackage{caption}
\makeatletter
\newif\if@restonecol
\makeatother

\usepackage[linesnumbered,ruled,vlined]{algorithm2e}
\SetKwRepeat{Loop}{do}{}

\allowdisplaybreaks[4]
\newtheorem{thm}{\textbf{Theorem}}

\newtheorem{definition}{\textbf{Definition}}
\newtheorem{lem}{\textbf{Lemma}}
\newtheorem{proposition}{\textbf{Proposition}}
\newtheorem{assumption}{\textbf{Assumption}}
\newtheorem{remark}{\textbf{Remark}}
\newtheorem{example}{\textbf{Example}}
\usepackage{tikz}
\usetikzlibrary{arrows,automata}
\usepackage{verbatim}
\usepackage{tikz}
\usepackage{abstract}

\usepackage{multicol}
\usepackage{multirow}
\usepackage{natbib}

\usepackage{array}

\begin{document}

\title{Zero-Sum Semi-Markov Games with State-Action-Dependent Discount Factors}
\author{{Zhihui Yu$^1$}, \ {Xianping Guo$^1$}\thanks{email: mcsgxp@mail.sysu.edu.cn}, \ {Li Xia$^2$}\thanks{email: xiali5@sysu.edu.cn; xial@tsinghua.edu.cn}\\
\small
    {$^1$School of Mathematics, Sun Yat-Sen University, Guangzhou, China}\\
\small \vspace{-1cm}
    {$^2$Business School, Sun Yat-Sen University, Guangzhou, China} }
\date{}
\maketitle
\begin{abstract}
    Semi-Markov model is one of the most general models for stochastic dynamic systems.
    This paper deals with a two-person zero-sum game for semi-Markov
    processes. We focus on the expected discounted payoff criterion with
    state-action-dependent discount factors. The state and action spaces are
    both Polish spaces, and the payoff function is $\omega$-bounded. We
    first construct a fairly general model of semi-Markov games
    under a given semi-Markov kernel and a pair of strategies. Next,
    based on the standard regularity condition and the
    continuity-compactness condition for semi-Markov games, we
    derive a ``drift condition" on the semi-Markov kernel and suppose
    that the discount factors have a positive lower bound, under which
    the existence of the value function  and a pair of optimal
    stationary strategies of our semi-Markov game are proved by using
    the Shapley equation. Moreover, when the state and action spaces are
    both finite, a value iteration-type algorithm for computing the
    value function and $\varepsilon$-Nash equilibrium of the game is
    developed. The convergence of the algorithm is also proved. Finally,
    we conduct numerical examples to demonstrate our main results.
\end{abstract}
\textbf{Keywords:} Semi-Markov game, state-action-dependent discount
factor, value  iteration-type algorithm, $\varepsilon$-Nash
equilibrium.

\section{Introduction}\label{sec1}
Game theory is a fundamental mathematical model to study strategic
interactions among rational decision-makers. It has wide
applications in many fields, such as social science, computer
science, management science, and economic systems. In the early
period of game theory, it focuses on \emph{matrix games} with two
persons and zero sum, where each participant's gains or losses are
exactly balanced by those of the other. When the system state
evolves over time, matrix games are transformed into
\emph{two-person zero-sum stochastic games}.

The study of zero-sum stochastic games is initiated by
\cite{Shapley1953}, and many extensions of that work have been
investigated in the literature. As is well known, it can be roughly
classified into the following four main groups. The first group is
\emph{discrete-time Markov games}
\citep{Hernandez-Lerma2000,Kuenle2003,Sennott1994}, which can be
considered as an extension of discrete-time Markov control
processes, that is, the decision epoch is the fixed discrete-time
point and the state-action process is discrete. The second group
deals with \emph{stochastic differential games}
\citep{Basar1999,Borkar1992,Kushner2004,Ramachandran1999}, where the
evolution of state variables is governed by stochastic differential
equations. The third group deals with \emph{continuous-time Markov
games} \citep{Guo2003,Guo2005,Guo2007,Neyman2017} in which the
sojourn times between consecutive decision epochs are exponentially
distributed and the players can select their actions continuously in
time. The fourth group is \emph{semi-Markov games} (SMGs)
\citep{Jaskiewicz2002,Lal1992,Luque-Vasquez2002,Minjarez2008,Mondal2016},
where the state process is continuous over time, the sojourn time
between two consecutive decision epochs follows any distribution and
players take actions just at the moment when the state changes.

In certain sense, we may argue that semi-Markov processes can model
almost every possible stochastic dynamic system, since the sojourn
time can be any distribution and the Markovian property can be
satisfied by state augment. Therefore, it is important to study
semi-Markov games which can be used to formulate wide varieties of
decision-making problems in social and economic systems. In this
paper, we focus on the study of two-person zero-sum SMGs.
\cite{Lal1992} deal with two-person zero-sum SMGs under both
expected discounted and long-run average payoff criterion, where the
state space is denumerable and the payoff function is bounded. For
the discounted case, they prove the existence of the value function
and a pair of optimal stationary strategies by using the
\emph{Shapley equation}. For the long-run average case, they further
establish the optimality equation and propose a standard ergodic
condition under which the existence of the value function and a pair
of optimal stationary strategies are ensured through solving the
optimality equation in a unified manner. \cite{Jaskiewicz2002}
studies the two-person zero-sum SMGs under the long-run average
payoff criterion with a more general model, where the state and
action spaces are both Borel spaces and the payoff function is
$\omega $-bounded. This paper derives some generalized geometric
ergodicity conditions on the transition probabilities under which
the optimality equation has a solution which can be obtained by
solving some $\varepsilon$-perturbed SGMs. This paper also proves
the existence of the value function and a pair of optimal stationary
strategies of the SMGs. There is further literature work on
two-person zero-sum SMGs that extends the similar results to the
expected discounted payoff criterion. \cite{Luque-Vasquez2002}
considers the $n$-stage SMGs as well as the infinite horizon case
with Borel state and action spaces and $\omega $-bounded payoff
function. The existence of the value function and a pair of optimal
stationary strategies are also shown under suitable assumptions on
the transition law. Moreover, \cite{Minjarez2008} study the
discounted zero-sum SMGs with unknown holding time distribution $H$
for one player. They propose a state-action independent condition on
$H$ to get independent observations during the evolution of the
system, under which they combine suitable methods of statistical
estimation of $H$ with control procedures to construct an
asymptotically discount optimal pair of strategies.

Most of the literature work on game theory focuses on the existence
of \emph{Nash equilibrium}. However, how to efficiently solve a
stochastic dynamic game and compute a pair of optimal stationary
strategies are especially important for practical implementation of
game theory. The classic algorithmic study on game theory focuses on
static games, where the matrix game and the bimatrix game can be
solved by linear programming and quadratic programming, respectively
\citep{Barron2013}. Recently, there are emerging investigations that
aim to study the efficient computation for stochastic dynamic games
using approximation or learning algorithms. \cite{Littman1994}
proposes a minimax-Q algorithm to solve discrete-time two-person
zero-sum Markov games, which is essentially motivated by the
standard Q-learning algorithm with a minimax operator in Markov
games replacing the max operator in reinforcement learning.
\cite{Al-Tamimi2007} utilize the method of Q-learning and
approximate dynamic programming (ADP) to solve a discrete-time
linear system quadratic zero-sum game, and the proof of the
convergence of the algorithm is also given. \cite{Vamvoudakis2012}
deal with a continuous-time two-person zero-sum game with infinite
horizon cost for nonlinear systems with known dynamics. They propose
a ``synchronous" zero-sum game policy iteration algorithm to solve
the game through learning the Hamilton-Jacobi-Isaacs (HJI) equation
in real time. Moreover, a persistence of excitation condition is
given under which the convergence to the optimal saddle point and
the stability of the system are also guaranteed. \cite{Mondal2016}
study the AR-AT (Additive Reward-Additive Transition) two-person
zero-sum SMGs, where the state and action spaces are both finite.
They prove that such game can be formulated as a vertical linear
complementarity problem (VLCP), which can be solved by the
Cottle-Dantzig's algorithm.

All the above literature work on SMGs assumes that the discount
factor is a constant, which may not always hold. For example,
considering the application in economics, the discount factor
(interest rate) may depend both on economy environments and
decision-makers' actions. That is, the interest rate usually varies
in different financial markets and monetary policies, where
financial markets can be considered as states and monetary policies
are actions taken by the government. Thus, it is necessary and
reasonable to study the SMGs with \emph{state-action-dependent
discount factors}. Problems with non-constant discount factors have
been studied for Markov decision processes (MDP)
\citep{Minjarez2015,Ye2012} and two-person zero-sum discrete-time
Markov games \citep{Gonzalez2019}. In this paper, we aim at studying
the two-person zero-sum SMGs with expected discounted payoff
criterion in which the discount factors are state-action-dependent.
The objective is to find a pair of optimal strategies to maximize
the payoff of player $1$ (P1) and minimize the payoff of player $2$
(P2). More precisely, we deal with the SMGs specified by five
primitive data: the state space $X$; the action spaces $A,B$ for P1
and P2, respectively; the semi-Markov kernel $Q(t,y|x,a,b)$; the
discount factor $\alpha(x,a,b)$; and the payoff function $r(x,a,b)$.
The state space $X$ and action spaces $A,B$ are all Polish spaces,
and the payoff function $r(x,a,b)$ is \emph{$\omega $-bounded}. With
these data, we construct an SMG model with a fairly general problem
setting. Then we impose suitable conditions on the model parameters
shown in Assumptions~\ref{ass1}-\ref{ass4}, under which we establish
the Shapley equation and prove the existence of the value function
and a pair of optimal stationary strategies of the game. Our proof
is quite different from \cite{Gonzalez2019} since we directly search
for Nash equilibrium in history-dependent strategies instead of
turning to Markov strategies. In addition, when the state and action
spaces are both finite, we derive a \emph{value iteration-type
algorithm} to approach to the value function and Nash equilibrium of
the game based on the Shapley equation. The convergence of the
algorithm is also proved. Finally, we conduct numerical examples on
investment problem to demonstrate the main results of our paper.

The contributions of this paper can be summarized as follows. (1) We
construct the two-person zero-sum SMG model with expected discounted payoff
criterion in which the discount factors are state-action-dependent. To the
best of our knowledge, our work is the first one that the discount
factor is regarded as a variable in stochastic semi-Markov games,
which could complement the theoretical study on SMGs. (2) We derive
a ``drift condition" (see Assumption~\ref{ass3}) on the semi-Markov
kernel, which is more general than the counterpart in the literature
work \citep{Luque-Vasquez2002}, as stated in Remark~\ref{re4}. (3)
We propose a value iteration-type algorithm to compute the value
function and $\varepsilon$-Nash equilibrium of the SMG. This
algorithm can be viewed as a combination of the value iteration of
MDP and the linear programming of
matrix games. Moreover, the convergence and the error-bound of the
algorithm are also guaranteed.

The rest of this paper is organized as follows. In
Section~\ref{sec2}, we introduce the model of SMG as well as the
optimality criterion. In Section~\ref{sec3}, we impose suitable
conditions on the model parameters under which the existence of the
value function and a pair of optimal stationary strategies are
proved by using the Shapley equation. A value iteration-type
algorithm for computing the $\varepsilon$-Nash equilibrium is
developed in Section~\ref{sec4}, and some numerical examples are
conducted to demonstrate our main results in Section~\ref{sec5}.
Finally, we conclude the paper and discuss some future research
topics in Section~\ref{sec6}.

\section{Two-Person Zero-Sum Semi-Markov Game Model}\label{sec2}

\noindent Notation: If $E$ is a Polish space (that is, a complete
and separable metric space), its Borel $\sigma$-algebra is denoted
by $\mathcal{B}(E)$, and  $\mathbb{P}(E)$ denotes the family of
probability measures on $\mathcal{B}(E)$ endowed with the topology
of weak convergence.

In this section, we introduce a two-person zero-sum SMG model with expected
discounted payoff criterion and state-action-dependent discount factors,
which is denoted by the collection
$$\{X, A, B, (A(x),B(x), x\in X),  Q(t,y|x,a,b), \alpha(x,a,b), r_{1}(x,a,b), r_{2}(x,a,b)\},$$
where the symbols are explained as follows.

\noindent$\bullet$ $X$ is the state space which is a Polish space,
and $A$ and $B$ are action spaces for P1 and P2, respectively, which
are also supposed to be Polish spaces.

\noindent$\bullet$ $A(x)$ and $B(x)$ are Borel subsets of $A$ and
$B$, which represent the sets of the admissible actions for P1 and
P2 at state $x\in X$, respectively. Let
\begin{equation*}
K := \{(x,a,b)|x\in X, a\in A(x), b\in B(x)\}
\end{equation*}
be a measurable Borel subset of $X \times A\times B$.

\noindent$\bullet$ $Q(t,y|x,a,b)$ is a semi-Markov kernel which
satisfies the following properties.

\noindent (a) For each fixed  $(x,a,b)\in K$, $Q(\cdot,\cdot|x,a,b)$
is a probability measure on $[0,+\infty) \times X$, whereas for each
fixed  $t \in [0,+\infty), D\in\mathcal{B}(X)$,
$Q(t,D|\cdot,\cdot,\cdot)$ is a real-valued Borel function on $K$.

\noindent (b) For each fixed $(x,a,b)\in K$ and
$D\in\mathcal{B}(X)$, $Q(\cdot,D|x,a,b)$ is a non-decreasing
right-continuous real-valued Borel function on $[0,+\infty)$ such
that $Q(0,D|x,a,b)=0$.

\noindent (c) For each fixed  $(x,a,b)\in K$,
\begin{equation*}
H(\cdot|x,a,b) := Q(\cdot,X|x,a,b)
\end{equation*}
denotes the distribution function of the sojourn time at state $x\in X$
when the actions $a\in A(x), b\in B(x)$ are chosen. For each  $x \in
X$ and $D\in\mathcal{B}(X)$, when P1 and P2 select actions $a \in
A(x)$ and $b \in B(x)$, respectively, $Q(t,D|x,a,b)$ denotes the
joint probability that the sojourn time in state $x$ is not greater
than $t \in R_{+}$ and the next state belongs to D.

\noindent$\bullet$ $\alpha(x,a,b)$ is a measurable function from $K$ to
$(0,+\infty)$ which denotes the state-action-dependent discount factor.

\noindent$\bullet$ $r_{1}(x,a,b)$ and $r_{2}(x,a,b)$ are two
real-valued functions on $K$, which represent the payoff function
for P1 and P2, respectively.

If $r_{1}(x,a,b)+r_{2}(x,a,b)=0$ for all $(x,a,b)\in K$, then the
model is called a two-person zero-sum SMG. Otherwise, the game is
nonzero-sum. In this paper, we focus on the zero-sum case. We denote
$r:=r_{1}=-r_{2}$, and regard P1 as the maximizer and P2 as the
minimizer. The evolution of SMGs with the expected discounted payoff
criterion carries on as follows.

Assume that the game starts at the initial state $x_{0}\in X$ at the
initial decision epoch $t_{0}:=0$. The two players choose
simultaneously a pure action pair $(a_{0},b_{0})\in A(x_{0})\times
B(x_{0})$ according to the variables $t_{0}$ and $x_{0}$, then P1
and P2 receive immediate rewards
$r_{1}(x_{0},a_{0},b_{0}),r_{2}(x_{0},a_{0},b_{0})$, respectively.
Consequently, after staying at state $x_{0}$ up to time
$t_{1}>t_{0}$, the system moves to a new state $x_{1}\in D$
according to the transition law
$Q(t_{1}-t_{0},D|x_{0},a_{0},b_{0})$. Once the state transition to
$x_{1}$ occurs at the $1$st decision epoch $t_{1}$, the entire
process repeats again and the game evolves in this way.

Thus, we obtain an admissible history at the $n$th decision epoch
\begin{equation}
h_{n}:=(t_{0},x_{0},a_{0},b_{0},t_{1},x_{1},a_{1},b_{1},\dots,
t_{n},x_{n}). \nonumber
\end{equation}
When the game goes to infinity, we obtain the history
\begin{equation}
h:=(t_{0},x_{0},a_{0},b_{0},t_{1},x_{1},a_{1},b_{1},\dots),\nonumber
\end{equation}
where $t_{n} \leq t_{n+1}$, $(x_{n},a_{n},b_{n})\in K$ for all $n
\geq 0$. Moreover, let $H_{n}$ be the class of all admissible
histories $h_{n}$ of the system up to the $n$th decision epoch,
endowed with a Borel $\sigma$-algebra.

To introduce our expected discounted payoff criterion discussed in this
paper, we give the definitions of strategies as follows.
\begin{definition}\label{defn1}
A randomized history-dependent strategy for P1 is a sequence of
stochastic kernels $\pi^{1}:=(\pi_{n}^{1},n\geq 0)$ that satisfies
the following conditions:

\noindent ($\textrm{i}$) for each $D\in \mathcal{B}(X)$,
$\pi_{n}^{1}(D|\cdot)$ is a Borel function on $H_{n}$, and for each
$h_{n}\in H_{n}$, $\pi_{n}^{1}(\cdot|h_{n})$ is a probability
measure on A;

\noindent ($\textrm{ii}$) $\pi_{n}^{1}(\cdot|h_{n})$ is concentrated
on $A(x_{n})$, that is
\begin{equation*}
\pi_{n}^{1}(A(x_{n})|h_{n})=1, ~~~ \forall h_{n}\in H_{n} ~
\mathrm{and}~ n\geq0.
\end{equation*}
\end{definition}
We denote by $\Pi_{1}$ the set of all the randomized
history-dependent strategies for P1 for simplicity.

\begin{definition}\label{defn2}
    (1) A strategy $\pi^{1}=(\pi_{n}^{1},n\geq 0) \in \Pi_{1}$ is called a randomized Markov
    strategy if there exists a sequence of stochastic kernels
    $\phi_{1}=(\varphi _{n},n\geq 0)$ such that
    \begin{equation*}
    \pi_{n}^{1}(\cdot|h_{n})=\varphi_{n}(\cdot|x_{n}),~~~ \forall h_{n}\in H_{n} ~ \mathrm{and}~ n\geq0.
    \end{equation*}

    (2) A randomized Markov strategy $\phi_{1}=(\varphi _{n},n\geq 0)$
    is called stationary if $\varphi _{n}$ is independent of $n$; that
    is, if there exists a stochastic kernel $\varphi$ on $A$ given $x$
    such that
    \begin{equation*}
    \varphi_{n}(\cdot|x)\equiv \varphi (\cdot|x),~~~ \forall x\in X~ \mathrm{and}~ n\geq0.
    \end{equation*}

    (3) Moreover, if $\varphi(\cdot|x)$ is a Dirac measure for all $x\in
    X$, then the stationary strategy
    $\varphi^\infty=(\varphi,\varphi,\varphi,\dots)$ is called a pure
    strategy.
\end{definition}
We denote by $\Pi_{1}^{M}$, $\Phi_{1}$ and $\Pi_{1}^{MD}$ the sets
of all the randomized Markov strategies, randomized stationary strategies and pure
strategies for P1, respectively.

The sets of all randomized history-dependent strategies $\Pi_{2}$,
randomized Markov strategies $\Pi_{2}^{M}$, randomized stationary
strategies $\Phi_{2}$, pure strategies $\Pi_{2}^{MD}$ for P2 are
defined similarly, with $B(x)$ in lieu of $A(x)$. Clearly,
$\Pi_{1}^{MD}\subset \Phi_{1} \subset \Pi_{1}^{M} \subset \Pi_{1} $
and $\Pi_{2}^{MD}\subset \Phi_{2} \subset \Pi_{2}^{M} \subset
\Pi_{2} $.


For each $x\in X,\pi^{1}\in \Pi_{1},\pi^{2}\in \Pi_{2}$,  by the
Tulcea's theorem \citep{Hernandez-Lerma1996}, there exists a unique
probability space $(\Omega, \mathcal{F}, \mathbb{P}_{x}^{\pi^{1},
\pi^{2}})$ and a stochastic process $\{T_{n}, X_{n}, A_{n}, B_{n},$
$ n\geq 0\}$ such that for each $D\in \mathcal{B}(X),D_{1}\in
\mathcal{B}(A),D_{2}\in \mathcal{B}(B)$ and $n\geq 0$, we have
\begin{equation*}
\mathbb{P}_{x}^{\pi^{1}, \pi^{2}}(X_{0}=x)=1,
\end{equation*}
\begin{equation*}
\mathbb{P}_{x}^{\pi^{1}, \pi^{2}}(A_{n}\in D_{1},B_{n}\in D_{2}|h_{n})=\pi_{n}^{1}(D_{1}|h_{n})\pi_{n}^{2}(D_{2}|h_{n}),
\end{equation*}
\begin{equation*}
\mathbb{P}_{x}^{\pi^{1}, \pi^{2}}(T_{n+1}-T_{n}\leq t ,X_{n+1}\in {D} | h_{n}, a_{n}, b_{n})=Q(t,D|x_{n},a_{n},b_{n}).
\end{equation*}
Corresponding to the stochastic process $\{T_{n},X_{n}, A_{n},
B_{n}, n\geq 0\}$ with probability space $(\Omega, \mathcal{F},
\mathbb{P}_{x}^{\pi^{1}, \pi^{2}})$, we define an underlying
continuous-time state-action process $\{X(t), A(t), B(t), t\geq 0\}$
as
\begin{equation*}
X(t)=\sum_{n=0}^{\infty}\mathbb I_{\{T_{n}\leq t<T_{n+1}\}}X_{n}+X^{c}\mathbb I_{\{t\geq T_{\infty}\}},
\end{equation*}
\begin{equation*}
A(t)=\sum_{n=0}^{\infty}\mathbb I_{\{T_{n}\leq t<T_{n+1}\}}A_{n}+A^{c}\mathbb I_{\{t\geq T_{\infty}\}},
\end{equation*}
\begin{equation*}
B(t)=\sum_{n=0}^{\infty}\mathbb I_{\{T_{n}\leq t<T_{n+1}\}}B_{n}+B^{c}\mathbb I_{\{t\geq T_{\infty}\}},
\end{equation*}
where $X^{c} \not\in X$, $A^{c} \not\in A$, $B^{c} \not\in B$ are
some isolated points, $T_{\infty}:=\lim\limits_{n\to+\infty}T_{n}$,
and $\mathbb I_{E}$ is an indicator function on any set $E$.
\begin{definition}\label{defn3}
    The stochastic process $\{X(t), A(t), B(t), t\geq 0\}$ is called a semi-Markov game.
\end{definition}

Next, we will show the definition of the expected discounted payoff criterion in this paper, where  $\mathbb{E}_{x}^{\pi^{1}, \pi^{2}}$ denotes the expectation operator associated with $\mathbb{P}_{x}^{\pi^{1}, \pi^{2}}$.

\begin{definition}\label{defn4}
    For each $(\pi^{1}, \pi^{2})\in \Pi_{1} \times \Pi_{2}$, the initial state $x \in X$ and discount factor $\alpha(\cdot)>0$, the expected discounted payoff criterion for
    player $i$ is defined as follows:
    \begin{equation}\label{equ1}
    V_{i}(x, \pi^{1}, \pi^{2}):= \mathbb{E}_{x}^{\pi^{1}, \pi^{2}}\Big[\int_{0}^{\infty} e^{-\int_{0}^{t}\alpha(X(s),A(s),B(s))ds}
    r_{i}(X(t), A(t), B(t))dt\Big], \quad i=1,2 .
    \end{equation}
\end{definition}

\begin{remark}\label{re1}
    Since $r=r_{1}=-r_{2}$, we just need to consider the expected discounted payoff criterion for
    P1. Let
    $$V(x, \pi^{1}, \pi^{2}):=V_{1}(x, \pi^{1}, \pi^{2}).$$
\end{remark}

This paper  focuses on  the study of the value function and Nash
equilibrium of the SMG. So we need the following concepts.

\begin{definition}\label{defn5}
    The upper value and lower value of the expected discounted payoff SMG are defined as
    $$U(x):= \inf_{\pi^{2}\in \Pi_{2}}\sup_{\pi^{1}\in \Pi_{1}}V(x, \pi^{1}, \pi^{2})~\mathrm{and}~L(x):= \sup_{\pi^{1}\in \Pi_{1}} \inf_{\pi^{2}\in \Pi_{2}}V(x, \pi^{1}, \pi^{2}),$$
    respectively. Obviously, $U(x)\ge L(x)$ for all $x \in X$. Moreover, if it holds that $L(x)= U(x)$ for all $x \in X$, then the common function is called the value function of the  game and denoted by $V^{*}$.
\end{definition}

\begin{definition}\label{defn6}
    Assume that the game has a value $V^{*}$. Then a strategy
    $\pi^{1}_{*}\in \Pi_{1}$ is said to be optimal for P1 if
    $$\inf_{\pi^{2}\in \Pi_{2}} V(x, \pi^{1}_{*}, \pi^{2})=V^{*}(x),~~
    \forall x\in X.$$ Similarly, $\pi^{2}_{*}\in \Pi_{2}$ is
    said to be optimal for P2 if
    $$\sup_{\pi^{1}\in \Pi_{1}} V(x, \pi^{1}, \pi^{2}_{*})=V^{*}(x),~~ \forall x\in X.$$
    If $\pi^{i}_{*}$ is optimal for player $i$ ($i=1,2$), then we can
    call $(\pi^{1}_{*}, \pi^{2}_{*})$  a pair of optimal strategies
    (Nash equilibrium).
\end{definition}

\begin{remark}\label{re2}
    $(\pi^{1}_{*}, \pi^{2}_{*})$ is a pair of optimal strategies if and only if
    \begin{equation*}
    V(x, \pi^{1}, \pi^{2}_{*})\leq V(x, \pi^{1}_{*}, \pi^{2}_{*})\leq V(x, \pi^{1}_{*}, \pi^{2}),~~~ \forall \pi^{1}\in \Pi_{1}, \pi^{2}\in \Pi_{2}.
    \end{equation*}
\end{remark}
Remark~\ref{re2} is an effective method to verify whether a pair of
strategy $(\pi^{1},\pi^{2})$ is a Nash equilibrium, which is widely
used in the literature; see, for instance,
\cite{Luque-Vasquez2002}, and the references therein.

\section{Optimality Analysis}\label{sec3}

In this section, we give some suitable assumptions on the model
parameters under which the existence of the value function and a
pair of  optimal stationary strategies are guaranteed. The related
proofs are also discussed.

Given a measurable function $\omega : X\rightarrow [1, \infty)$, a function $u$ on $X$ is said to be $\omega$-bounded if it has finite $\omega$-norm which is defined as
$$\|u\|_{\omega}:=\sup_{x\in X}\frac{|u(x)|}{\omega(x)},$$
such a function $\omega$ can be referred to as a weight function. For convenience, we write $B_{\omega}(X)$ the Banach space of all $\omega$-bounded measurable functions on $X$.

Next, we give some hypotheses to guarantee the existence of a pair
of optimal strategies.
\begin{assumption}\label{ass1}
    There exist constants $\theta>0$ and $\delta>0$ such that $$H(\theta| x, a, b) \leq 1-\delta, \quad \forall(x, a, b) \in K.$$
\end{assumption}

\begin{remark}\label{re3}
    Assumption~\ref{ass1} is a regularity condition which indicates that for each fixed $x\in X$ and $\pi^{1}\in \Pi_{1},\pi^{2}\in \Pi_{2}$, we have
$$\mathbb{P}_{x}^{\pi^{1}, \pi^{2}}(\lim\limits_{n\to+\infty}T_{n}=+\infty )=1,$$
    which avoids possibility of infinitely many decision epochs during the finite time interval; see, for instance,  \cite{Lal1992},  \cite{Luque-Vasquez2002}, and the references therein .
\end{remark}

To guarantee the finiteness of the expected discounted payoff
defined in (\ref{equ1}), we propose the following assumption.
\begin{assumption}\label{ass2}
    (a) There exists a constant $\alpha_{0}>0$ such that $\alpha(x,a,b) \geq \alpha_{0}$ for all $ (x,a,b) \in K$.
    (b) There exists a measurable function $\omega : X\rightarrow [1, \infty)$ and a nonnegative constant $M$ such that for all $(x,a,b)\in K$,
\begin{equation*}
    |r(x,a,b)|\leq M\omega(x).
\end{equation*}
\end{assumption}

Below we give an important consequence of Assumption~\ref{ass1} and
Assumption~\ref{ass2}(a).
\begin{lem}\label{lem1}
    If Assumptions~\ref{ass1}\&\ref{ass2}(a) hold, then there exists a
    constant $0<\gamma<1$ such that for each $(x, a, b)
    \in K$,
\begin{equation}\label{equ2}
    \int_{0}^{\infty} e^{-\alpha(x,a,b) t} H(d t|x, a, b) \leqslant \gamma
\end{equation}
    \begin{proof}
        For each fixed $(x, a, b) \in K$, integrating by parts and we have
\begin{align*}
        \int_{0}^{\infty} e^{-\alpha(x,a,b) t} H(d t | x, a, b)
        &=\alpha(x,a,b) \int_{0}^{\infty} e^{-\alpha(x,a,b) t} H(t | x, a, b) d t \\
        &=\alpha(x,a,b)\Big[\int_{0}^{\theta} e^{-\alpha(x,a,b) t} H(t | x, a, b) d t+\int_{\theta}^{\infty} e^{-\alpha(x,a,b) t} H(t |x, a, b) d t\Big] \\
        &\leq \alpha(x,a,b)\Big[(1-\delta)\int_{0}^{\theta}  e^{-\alpha(x,a,b) t} d t+\int_{\theta}^{\infty} e^{-\alpha(x,a,b) t} dt\Big]  \\
        &=1-\delta \left(1-e^{-\alpha(x,a,b)  \theta}\right)\\
        &\leqslant 1-\delta+\delta e^{-\alpha_{0}\theta}<1.
\end{align*}
        Let $\gamma=1-\delta+\delta e^{-\alpha_{0}\theta}$,
        which yields (\ref{equ2}).
    \end{proof}
\end{lem}

\begin{assumption}\label{ass3}
    There exists a constant $\eta$ with $0<\eta\gamma<1$ such that for all fixed $t\geq 0$ and $(x,a,b)\in K$,
\begin{equation}\label{equ3}
    \int_{X} \omega(y) Q(t,dy|x,a,b) \leq
    \eta\omega(x)H(t|x,a,b),
\end{equation}
    where $\omega(\cdot)$ is the function mentioned in Assumption~\ref{ass2}.
\end{assumption}

\begin{remark}\label{re4}
$(1)$  We call Assumption~\ref{ass3}  the ``drift condition", which
is needed to ensure that the Shapley operator (defined later in
(\ref{equ5})) is a contraction operator as well as our main results.
    Particularly, if $Q(t,y|x,a,b)=H(t|x,a,b)P(y|x,a,b)$, where $P(y|x,a,b)$ denotes the state transition probability,
    (\ref{equ3}) degenerates into  $\int_{X} \omega(y) P(dy|x,a,b) \leq
    \eta\omega(x)$, which is the same as the Assumption~\ref{ass3}(b) of
    \cite{Luque-Vasquez2002}. Thus, our Assumption~\ref{ass3} is
    more general than the counterpart in the literature \cite{Luque-Vasquez2002}.

$(2)$ Combining Lemma~\ref{lem1} with  Assumption~\ref{ass3}, it is easy to derive
\begin{equation}\label{equ30}
\int_{0}^{\infty} e^{-\alpha(x,a,b) t}\int_{X} u(y) Q(d t, d y | x, a, b)\leq \eta\gamma\|u\|_{\omega}\omega(x),\quad\forall u\in B_{\omega}(X),(x,a,b)\in K.
\end{equation}
\end{remark}
Moreover, we impose the following continuity-compactness conditions
to ensure the existence of a pair of optimal stationary strategies of our SMG model.
\begin{assumption}\label{ass4}
    (a) For each fixed $x\in X$, $A(x)$ and $B(x)$ are both compact sets.

    (b) For each fixed  $(x, a, b)\in K$, $r(x,\cdot,b)$ is upper semi-continuous  on $A(x)$ and $r(x,a,\cdot)$ is lower semi-continuous  on $B(x)$.

    (c) For each fixed $(x, a, b)\in K$, $t\geq 0$ and $v\in B_{\omega}(X)$, the functions
$$a \longmapsto \int v(y) Q(t,d y | x, a, b) \quad \text { and } \quad b \longmapsto \int v(y) Q(t,d y | x, a, b)$$
are continuous on $A(x)$ and $B(x)$, respectively.

    (d) For each fixed  $t\geq 0$, $H(t|\cdot,\cdot,\cdot)$ is continuous on $K$.

    (e) The function $\alpha(x,a,b)$ is continuous on $K$.
\end{assumption}

\begin{remark}\label{re5}
    $(1)$ Assumption~\ref{ass4} is similar to the standard  continuity-compactness hypotheses for Markov control processes; see, for instance,  \cite{Hernandez-Lerma1999}, and the references therein. It is commonly used for the existence of minmax points of  games.

    $(2)$ By Lemma~$1.11$ in  \cite{Nowak1984}, if Assumption~\ref{ass4}(a) holds, then the probability spaces $\mathbb{A}(x):=\mathbb{P}(A(x)) \text { and } \mathbb{B}(x):=\mathbb{P}(B(x))$ are also compact for each $x\in X$.
\end{remark}

We now introduce the following notations: for each given function $u \in B_{w}(X)$ and $(x,a,b)\in K$, we write
\begin{equation}\label{equ4}
G(u, x, a, b):=r(x, a, b)\int_{0}^{\infty} e^{-\alpha(x,a,b) t}(1-H(t| x, a, b))dt+\int_{0}^{\infty} e^{-\alpha(x,a,b) t}\int_{X} u(y) Q(d t, d y | x, a, b).
\end{equation}
For each fixed $x \in X$ and probability measures $\mu \in \mathbb{A}(x) \text { and } \lambda \in \mathbb{B}(x)$, we denote
\begin{equation*}
G(u, x, \mu, \lambda):=\int_{A(x)} \int_{B(x)} G(u, x, a, b) \mu(d a) \lambda(d b),
\end{equation*}
whenever the integral is   well defined.

We define an operator $T$ on $B_{\omega}(X)$ by
\begin{equation}\label{equ5}
Tu(x):=\sup _{\mu \in \mathbb{A}(x)} \inf _{\lambda \in \mathbb{B}(x)} G(u, x, \mu, \lambda),\quad  \forall x \in X,
\end{equation}
which is called the Shapley operator. A function $v \in
B_{w}(X)$ is said to be a solution of the Shapley equation if
$$Tv(x)=v(x), \quad  \forall x \in X.$$

In order to explore the existence of a pair of optimal stationary
strategies, we also need to define another operator $T(f,g)$ on
$B_{w}(X)$ by
$$T(f,g)u(x):=G(u, x, f(x), g(x)),\quad  \forall x \in X,$$
where $(f, g) \in \Phi_{1} \times \Phi_{2}$ is a pair of stationary strategies.

Before stating our main results, we need the following lemmas:

\begin{lem}\label{lem2}
    Suppose that Assumptions~\ref{ass1}-\ref{ass4} hold, then for each given function $u \in B_{\omega}(X)$, the function $Tu$ is in $B_{\omega}(X)$ and
\begin{equation}\label{equ6}
    Tu(x):=\min _{\lambda \in \mathbb{B}(x)} \max _{\mu \in \mathbb{A}(x)} G(u,x, \mu,
    \lambda).
\end{equation}
    Moreover, there exists a pair of stationary strategies $(f, g) \in \Phi_{1} \times \Phi_{2}$ such that
\begin{equation}\label{equ7}
\begin{aligned}
    T u(x) &=G(u, x, f(x), g(x)) \\
    &=\max _{\mu \in \mathbb{A}(x)} G(u, x, \mu, g(x)) \\
    &=\min _{\lambda \in \mathbb{B}(x)} G(u, x, f(x), \lambda).
\end{aligned}
\end{equation}
\end{lem}
\begin{proof}

    By Assumption~\ref{ass2} and formulation (\ref{equ30}), for each given function $u \in B_{w}(X)$ and
    $(x,a,b)\in K$, we can easily get
\begin{equation*}
    |G(u, x, a, b)|\leq\frac{M}{\alpha_{0}}\cdot\omega(x)+\eta\gamma\|u\|_{w} \cdot
    \omega(x).
\end{equation*}
    The above inequality yields $\|G(u, \cdot, a, b)\|_{\omega}\leq
    \frac{M}{\alpha_{0}}+\eta\gamma\|u\|_{w}$, which implies
    $G(u, x, a, b)$ is in $B_{\omega}(X)$, and so $Tu \in
    B_{\omega}(X)$.

    On the one hand, by Assumption~\ref{ass4}, it follows that
    $G(u,x,\cdot,b)$ is upper semi-continuous in $A(x)$, then for each
    fixed $\lambda \in \mathbb{B}(x)$, by the Fatou's theorem, the
    function
$$a \longmapsto \int_{B(x)} G(u, x, a, b) \lambda(d b)$$
    is also upper semi-continuous in $A(x)$. Moreover, since the
    probability measures on $\mathcal{B}(X)$ endowed with the topology
    of weak convergence, by Theorem~$2.8.1$ in \cite{Ash2000}, the
    function $G(u, x, \cdot, \lambda)$ is upper semi-continuous in
    $\mathbb{A}(x)$. Similarly, $G(u, x, \mu, \cdot)$ is lower
    semi-continuous in $\mathbb{B}(x)$. Thus, by Theorem~$A.2.3$ in
    \cite{Ash2000}, the supremum and the infimum are indeed attained in
    (\ref{equ5}), which means
$$Tu(x)=\max _{\lambda \in \mathbb{B}(x)} \min _{\mu \in \mathbb{A}(x)} G(u, x, \mu, \lambda).$$
    Then, by the Fan's minimax Theorem \citep{Fan1953}, we obtain
    (\ref{equ6}).

    On the other hand, it is clear that $G(x, u, \mu, \lambda)$ is both concave and convex in $\mathbb{A}(x)$ with respect to $\mu$ and in $\mathbb{B}(x)$ with respect to $\lambda$. Hence, by the well-known measurable selection theorem  \citep{Nowak1985}, there exists a pair of stationary strategies $(f, g) \in \Phi_{1} \times \Phi_{2}$  that satisfies (\ref{equ7}).
\end{proof}

\begin{lem}\label{lem3}
    Both $T$ and $T(f,g)$  are contraction operators with modulus less than $1$.
\end{lem}
\begin{proof}
    First, it is easy to verify that the operator $T(f,g)$ is monotonically increasing. Let $u,v \in B_{\omega}(X)$, by the definition of $\omega$-norm, $u(\cdot) \leq v(\cdot)+\|u-v\|_{\omega}\omega(\cdot)$, it follows that for  each fixed $x\in
    X$, we have
\begin{equation}\label{equ8}
\begin{aligned}
    T(f, g) u(x) & \leqslant T(f, g)(v+\omega \|u-v\|_{\omega})(x) \\
    &=T(f,g)v(x)\\
    &+ \|u-v\|_{\omega}  \int_{A(x)} \int_{B(x)}\Big[\int_{0}^{\infty} e^{-\alpha(x,a,b) t} \int_{X}\omega (y)Q(d t, d y| x, a, b)\Big] f(da|x) g(db|x)\\
    & \leqslant T(f,g)v(x)+\eta\gamma \|u-v\|_{\omega}\omega (x),
\end{aligned}
\end{equation}
    where the last inequality is followed by formulation (\ref{equ30}).
    Furthermore, taking maximum of $f\in \Phi_{1}$ and minimum of $g\in
    \Phi_{2}$ on both sides of the inequality (\ref{equ8}), we have
    $$\max _{f \in \Phi_{1}} \min _{g\in \Phi_{2}}T(f,g)u(x)\leq \max _{f \in \Phi_{1}} \min _{g\in \Phi_{2}}T(f,g)v(x)+\eta\gamma \|u-v\|_{\omega}\omega (x),$$
    i.e.
    $$Tu(x)\leq Tv(x)+\eta\gamma \|u-v\|_{\omega}\omega (x).$$
    Similarly, interchanging $u$ and $v$, we obtain
    $$Tv(x)\leq Tu(x)+\eta\gamma \|v-u\|_{\omega}\omega (x).$$
    Combining the two inequalities above, we have
    $$|Tu(x)-Tv(x)|\leq \eta\gamma \|u-v\|_{\omega}\omega (x) ,\quad \forall x\in X,$$
    i.e.
    $$\|Tu-Tv\|_{\omega} \leq \eta\gamma \|u-v\|_{\omega},$$
    which implies $T$ is a contraction operator with modulus
    $\eta\gamma<1$. Using the same arguments, we can prove that
    $T(f,g)$ is also a contraction operator with modulus
    $\eta\gamma<1$.
\end{proof}

Since $T$ and $T(f,g)$ are both contraction operators, then there
exist unique functions $u^{*}\in B_{\omega}(X)$ and $u_{f,g}^{*}\in
B_{\omega}(X)$ such that $Tu^{*}(\cdot)=u^{*}(\cdot)$ and
$T(f,g)u_{f,g}^{*}(\cdot)=u_{f,g}^{*}(\cdot)$ by the Banach's fixed
point theorem.

\begin{lem}\label{lem4}
    For each $(\pi^{1}, \pi^{2})\in \Pi_{1} \times \Pi_{2}$ and $x \in X$,
    $$V(x,\pi^{1}, \pi^{2})=T(\pi_{0}^{1}, \pi_{0}^{2})V(x, ^{(1)}\!\pi^{1},  ^{(1)}\!\pi^{2}),$$
    where $\pi^{i}:=(\pi_{n}^{i},n\geq 0)$, and $^{(1)}\!\pi^{i}:=(\pi_{n}^{i},n\geq 1)$ which denotes the translation of strategy.
\end{lem}
\begin{proof}
    \begin{align*}
    {V(x, \pi^{1}, \pi^{2})}&= \mathbb{E}_{x}^{\pi^{1}, \pi^{2}}\Big[\int_{0}^{\infty}e^{-\int_{0}^{t}\alpha(X(s),A(s),B(s))ds}r(X(t), A(t), B(t))dt\Big]\\
    &= \mathbb{E}_{x}^{\pi^{1},\pi^{2}}\Big[\int_{0}^{T_{1}} e^{-\int_{0}^{t}\alpha(X(s),A(s),B(s))ds} r(X(t), A(t), B(t))dt\Big]\\
    &+ \mathbb{E}_{x}^{\pi^{1},\pi^{2}}\Big[\int_{T_{1}}^{\infty} e^{-\int_{0}^{t}\alpha(X(s),A(s),B(s))ds} r(X(t), A(t), B(t))dt\Big]\\
    &=
    \mathbb{E}_{x}^{\pi^{1},\pi^{2}}\Big[\int_{0}^{\infty}\mathbbm{1}_{\{T_{1}>t\}} e^{-\alpha(x,A_{0},B_{0})t} r(x, A_{0},B_{0})dt\Big]\\
    &+
    \mathbb{E}_{x}^{\pi^{1},\pi^{2}}\bigg[  \mathbb{E}_{x}^{\pi^{1},\pi^{2}}\Big[\int_{T_{1}}^{\infty}e^{-\alpha(x,A_{0},B_{0})T_{1}} e^{-\int_{T_{1}}^{t}\alpha(X(s),A(s),B(s))ds} r(X(t), A(t), B(t))dt|h_{1}\Big]\bigg]\\
    &=
    \int_{A(x)}\int_{B(x)}\Big[\int_{0}^{\infty}e^{-\alpha(x,a,b)t}\big[1-H(t| x, a, b)\big]r(x,a,b)dt\Big] \pi_{0}^{1}(da|x) \pi_{0}^{2}(db|x)\\
    &+
    \mathbb{E}_{x}^{\pi^{1},\pi^{2}}\bigg[e^{-\alpha(x,A_{0},B_{0})T_{1}}   \mathbb{E}_{x}^{\pi^{1},\pi^{2}}\Big[\int_{T_{1}}^{\infty} e^{-\int_{T_{1}}^{t}\alpha(X(s),A(s),B(s))ds} r(X(t), A(t), B(t))dt|h_{1}\Big]\bigg]\\
    &=
    \int_{A(x)}\int_{B(x)}\Big[\int_{0}^{\infty}e^{-\alpha(x,a,b)t}\big[1-H(t| x, a, b)\big]r(x,a,b)dt\Big] \pi_{0}^{1}(da|x) \pi_{0}^{2}(db|x)\\
    &+
    \mathbb{E}_{x}^{\pi^{1},\pi^{2}}\Big[e^{-\alpha(x,A_{0},B_{0})T_{1}}V(x(T_{1}),^{(1)}\!\pi^{1},^{(1)}\!\pi^{2})\Big]\\
    &=
    \int_{A(x)}\int_{B(x)}r(x,a,b)\Big[\int_{0}^{\infty}e^{-\alpha(x,a,b)t}\big[1-H( t|x, a, b)\big]dt\Big] \pi_{0}^{1}(da|x) \pi_{0}^{2}(db|x)\\
    &+
    \int_{A(x)}\int_{B(x)}\Big [\int_{0}^{\infty}e^{-\alpha(x,a,b)t}\int_{X} V(y,^{(1)}\!\pi^{1},^{(1)}\!\pi^{2})Q(d t,dy| x, a, b)\Big]\pi_{0}^{1}(da|x) \pi_{0}^{2}(db|x)\\
    &=
    \int_{A(x)}\int_{B(x)}\left\{r(x,a,b)\Big[\int_{0}^{\infty}e^{-\alpha(x,a,b)t}\big[1-H(t|x, a, b)\big]dt\Big]+\right.\\
    &\phantom{=\;\;}\left.
    \int_{0}^{\infty}e^{-\alpha(x,a,b)t}\Big [\int_{X} V(y,^{(1)}\!\pi^{1},^{(1)}\!\pi^{2})Q(d t,dy| x, a, b)\Big] \right\}\pi_{0}^{1}(da|x) \pi_{0}^{2}(db|x),
    \end{align*}
    where the third and fourth equalities are ensured by the property of conditional expectation. The
    fifth equality follows from the strong Markovian property.
    Hence, $$V(x,\pi^{1}, \pi^{2})=T(\pi_{0}^{1}, \pi_{0}^{2})V(x,^{(1)}\!\pi^{1},^{(1)}\!\pi^{2}),$$
    which is required.
\end{proof}

Now, if we set $\pi^{1}=f$ and $\pi^{2}=g$ specially, which are both
stationary strategies, from Lemma~\ref{lem4}, we have
\begin{equation*}
V(x,f,g)=T(f,g)V(x,f,g),~~~~\forall x\in X,
\end{equation*}
which implies that the function $V(x,f,g)$ is the unique fixed point
of the contraction operator $T(f,g)$.

\begin{lem}\label{lem5}
    Suppose that Assumptions~\ref{ass1}-\ref{ass3} hold, let $(\pi^{1}, \pi^{2})\in \Pi_{1} \times \Pi_{2}$, then for each $x\in X$, $u\in B_{\omega}(X)$, we have
    $$\lim\limits_{n\to+\infty}\mathbb{E}_{x}^{\pi^{1}, \pi^{2}}\Big[e^{-\int_{0}^{T_{n}}\alpha(X(s),A(s),B(s))ds}u(X_{n})\Big]=0$$
\end{lem}
\begin{proof}
    For $\forall n\ge1$ and $x\in X$, we have
    \begin{align*}
    &\bigg|\mathbb{E}_{x}^{\pi^{1}, \pi^{2}}\Big[e^{-\int_{0}^{T_{n}}\alpha(X(s),A(s),B(s))ds}\omega(X_{n})\Big]\bigg|\\
    &= \bigg|\mathbb{E}_{x}^{\pi^{1},\pi^{2}}\Big[\mathbb{E}_{x}^{\pi^{1}, \pi^{2}}\big[e^{-\int_{0}^{T_{n}}\alpha(X(s),A(s),B(s))ds}\omega(X_{n})|h_{n-1},A_{n-1},B_{n-1}\big]\Big]\bigg|\\
    &= \bigg|\mathbb{E}_{x}^{\pi^{1},\pi^{2}}\Big[e^{-\int_{0}^{T_{n-1}}\alpha(X(s),A(s),B(s))ds}\mathbb{E}_{x}^{\pi^{1}, \pi^{2}}\big[e^{-\int_{T_{n-1}}^{T_{n}}\alpha(X(s),A(s),B(s))ds}\omega(X_{n})|h_{n-1},A_{n-1},B_{n-1}\big]\Big]\bigg|\\
    &=
    \bigg|\mathbb{E}_{x}^{\pi^{1},\pi^{2}}\Big[e^{-\int_{0}^{T_{n-1}}\alpha(X(s),A(s),B(s))ds}\big[\int_{0}^{\infty}e^{-\alpha(X_{n-1},A_{n-1},B_{n-1})t}\int_{X}\omega(y)Q(d t,dy| X_{n-1}, A_{n-1},B_{n-1})\big]\Big]\bigg|\\
    &\leq
    \eta\gamma\bigg|\mathbb{E}_{x}^{\pi^{1},\pi^{2}}\Big[e^{-\int_{0}^{T_{n-1}}\alpha(X(s),A(s),B(s))ds}\omega(X_{n-1})\Big]\bigg|,
    \end{align*}
    where the first and second equalities are ensured by the property of conditional expectation. The last
    inequality follows from formulation (\ref{equ30}).
    Through iteration we have
    \begin{align*}
    \bigg|\mathbb{E}_{x}^{\pi^{1}, \pi^{2}}\Big[e^{-\int_{0}^{T_{n}}\alpha(X(s),A(s),B(s))ds}u(X_{n})\Big]\bigg|&\leq \|u\|_{\omega}\bigg|\mathbb{E}_{x}^{\pi^{1}, \pi^{2}}\Big[e^{-\int_{0}^{T_{n}}\alpha(X(s),A(s),B(s))ds}\omega(X_{n})\Big]\bigg|\\
    &\leq (\eta\gamma)^n\|u\|_{\omega}\omega(x),
    \end{align*}
    which yields Lemma~\ref{lem5}.
\end{proof}
Next, we present our main results.

\begin{thm}\label{thm1}
    Suppose that Assumptions~\ref{ass1}-\ref{ass4} hold, then

    (a) The semi-Markov game has a value function $V^{*}(\cdot)$, which is the unique function in $B_{\omega}(X)$ that satisfies the Shapley equation
    $$V^{*}(x)=TV^{*}(x),\quad\forall x\in X,$$
    and furthermore, there exists a pair of optimal strategies.

    (b) A pair of stationary strategies $(f^{*}, g^{*}) \in \Phi_{1} \times \Phi_{2}$ is optimal if and only if its expected payoff satisfies the Shapley equation $TV(x,f^{*},g^{*})=V(x,f^{*},g^{*})$ for all $x\in X$.
\end{thm}
\begin{proof}
    (a) Let $u^{*}$ be the unique fixed point of $T$ in $B_{\omega}(X)$, that is
    $$u^{*}(x)=Tu^{*}(x),~~~~\forall x\in X.$$
    By Lemma~\ref{lem2}, there exists a pair of stationary strategies $(f_{1}^{*}, g_{1}^{*}) \in \Phi_{1} \times \Phi_{2}$ such that for each $x\in X$,
    \begin{equation}\label{equ9}
    \begin{aligned}
    Tu^{*}(x) &=G(u^{*}, x, f_{1}^{*}(x), g_{1}^{*}(x)) \\
    &=\max _{\mu \in \mathbb{A}(x)} G(u^{*}, x, \mu, g_{1}^{*}(x)) \\
    &=\min _{\lambda \in \mathbb{B}(x)} G(u^{*}, x, f_{1}^{*}(x), \lambda) ,
    \end{aligned}
    \end{equation}
    which implies that $$u^{*}(x)=G(u^{*}, x, f_{1}^{*}(x),
    g_{1}^{*}(x))=T(f_{1}^{*},g_{1}^{*})u^{*}(x),\quad \forall x\in X.$$ Moreover, by
    Lemma~\ref{lem4},
    $$V(x,f_{1}^{*},g_{1}^{*})=T(f_{1}^{*},g_{1}^{*})V(x,f_{1}^{*},g_{1}^{*}),\quad \forall x\in X,$$
    from which we can derive $$u^{*}(x)=V(x,f_{1}^{*},g_{1}^{*}),\quad \forall x\in X.$$
    Next, we prove that $u^{*}$ is the value function of the game and
    $(f_{1}^{*},g_{1}^{*})$ is a pair of optimal strategies, that is
    \begin{equation}\label{equ10}
    V(x,f_{1}^{*},\pi^{2})\geq V(x,f_{1}^{*},g_{1}^{*})\geq
    V(x,\pi^{1},g_{1}^{*})~,~~~\forall (\pi^{1},\pi^{2})\in \Pi_{1}
    \times \Pi_{2},
    \end{equation}

    We first prove the first inequality in (\ref{equ10}). Then a similar
    proof can follow for the second inequality. By (\ref{equ9}), we
    have
    $$u^{*}(x)\leq G(u^{*},x,f_{1}^{*}(x),\lambda)~,~~~\forall
    \lambda \in \mathbb{B}(x).$$
    Particularly, let $\lambda$ be an
    indicator function such that $\lambda (db)=1$. Then for each $b\in
    B(x)$, we have
    \begin{equation*}
    \begin{aligned}
    u^{*}(x)&\leq \int_{A(x)}\left\{r(x,a,b)\int_{0}^{\infty}e^{-\alpha(x,a,b)t}\Big[1-H(t| x, a, b)\Big]dt+\right.\\
    &\phantom{=\;\;}\left.\int_{0}^{\infty}e^{-\alpha(x,a,b)t}\Big [\int_{X}u^{*}(y)Q(d t,dy| x, a, b)\Big] \right\}f_{1}^{*}(da|x),
    \end{aligned}
    \end{equation*}
    taking $x$ as a random variable $X_{n}$, then for all $b_{n}\in
    B(X_{n})$, we have
    \begin{equation*}
    \begin{aligned}
    u^{*}(X_{n})&\leq \int_{A(X_{n})}\left\{r(X_{n},a_{n},b_{n})\int_{0}^{\infty}e^{-\alpha(X_{n},a_{n},b_{n})t}\Big[1-H(t| X_{n},a_{n},b_{n})\Big]dt+\right.\\
    &\phantom{=\;\;}\left.\int_{0}^{\infty}e^{-\alpha(X_{n},a_{n},b_{n})t}\Big
    [\int_{X}u^{*}(y)Q(d t,dy|X_{n},a_{n},b_{n})\Big]
    \right\}f_{1}^{*}(da_{n}|X_{n}).
    \end{aligned}
    \end{equation*}
    For $\forall \pi^2\in \Pi_{2}$, integrating $b_{n}$ on both sides in the above inequality, we have
    \begin{equation*}
    \begin{aligned}
    u^{*}(X_{n})&\leq \int_{B(X_{n})}\int_{A(X_{n})}\left\{\int_{0}^{\infty}e^{-\alpha(X_{n},a_{n},b_{n})t}\Big [\int_{X}u^{*}(y)Q(dt,dy|X_{n},a_{n},b_{n})\Big]+\right. \\
    &\phantom{=\;\;}\left.r(X_{n},a_{n},b_{n})\int_{0}^{\infty}e^{-\alpha(X_{n},a_{n},b_{n})t}\Big[1-H(t| X_{n},a_{n},b_{n})\Big]dt \right\}f_{1}^{*}(da_{n}|X_{n})\pi_{n}^{2}(db_{n}|h_{n})\\
    &=\mathbb{E}_{x}^{f_{1}^{*},
        \pi^{2}}\Big[e^{-\int_{T_{n}}^{T_{n+1}}\alpha(X(s),A(s),B(s))ds}u^{*}(X_{n+1})|h_{n}\Big]+\\
    &~~~~\mathbb{E}_{x}^{f_{1}^{*}, \pi^{2}}\Big[\int_{T_{n}}^{T_{n+1}}
    e^{-\int_{T_{n}}^{t}\alpha(X(s),A(s),B(s))ds} r(X(t), A(t),
    B(t))dt|h_{n}\Big].
    \end{aligned}
    \end{equation*}
    Multiplying $e^{-\int_{0}^{T_{n}}\alpha (x(s))ds}$ on both sides in
    the above inequality and using the properties of the conditional
    expectation, we have
    \begin{equation*}
    \begin{aligned}
    e^{-\int_{0}^{T_{n}}\alpha(X(s),A(s),B(s))ds}u^{*}(X_{n})&\leq
    \mathbb{E}_{x}^{f_{1}^{*}, \pi^{2}}\Big[e^{-\int_{0}^{T_{n+1}}\alpha(X(s),A(s),B(s))ds}u^{*}(X_{n+1})|h_{n}\Big]+\\
    &\mathbb{E}_{x}^{f_{1}^{*}, \pi^{2}}\Big[\int_{T_{n}}^{T_{n+1}}
    e^{-\int_{0}^{t}\alpha(X(s),A(s),B(s))ds} r(X(t), A(t), B(t))dt|h_{n}\Big].
    \end{aligned}
    \end{equation*}
    Then, taking the expectation $\mathbb{E}_{x}^{f_{1}^*, \pi^{2}}$, we
    have
    \begin{equation*}
    \begin{aligned}
    \mathbb{E}_{x}^{f_{1}^{*}, \pi^{2}}\Big[e^{-\int_{0}^{T_{n}}\alpha(X(s),A(s),B(s))ds}u^{*}(X_{n})\Big]&\leq
    \mathbb{E}_{x}^{f_{1}^{*}, \pi^{2}}\Big[e^{-\int_{0}^{T_{n+1}}\alpha(X(s),A(s),B(s))ds}u^{*}(X_{n+1})\Big] \\
    &+\mathbb{E}_{x}^{f_{1}^{*}, \pi^{2}}\Big[\int_{T_{n}}^{T_{n+1}} e^{-\int_{0}^{t}\alpha(X(s),A(s),B(s))ds}
    r(X(t), A(t), B(t))dt\Big].
    \end{aligned}
    \end{equation*}
    Now, summing over $n=0,1,2,\dots,N$, we obtain
    \begin{align*}
    u^*(x)&\leq \mathbb{E}_{x}^{f_{1}^{*}, \pi^{2}}\Big[\int_{0}^{T_{N+1}} e^{-\int_{0}^{t}\alpha(X(s),A(s),B(s))ds}r(X(t), A(t), B(t))dt\Big]\\
    &+\mathbb{E}_{x}^{f_{1}^{*}, \pi^{2}}\Big[e^{-\int_{0}^{T_{N+1}}\alpha(X(s),A(s),B(s))ds}u^{*}(X_{N+1})\Big].
    \end{align*}
    Letting $N\rightarrow+\infty$, according to Lemma~\ref{lem5}, we derive
    $$u^*(x)\leq \mathbb{E}_{x}^{f_{1}^{*}, \pi^{2}}\Big[\int_{0}^{\infty} e^{-\int_{0}^{t}\alpha(X(s),A(s),B(s))ds}r(X(t), A(t), B(t))dt\Big],$$
    which means that the first inequality in (\ref{equ10}) holds.

    (b) ($\Rightarrow$)

    Suppose that $(f^{*}, g^{*}) \in \Phi_{1} \times \Phi_{2}$ is a pair of  optimal stationary strategies, then for each $x\in X,\pi^{1} \in  \Pi_{1},\pi^{2} \in  \Pi_{2}$, we have
    $$V(x,f^{*},\pi^{2})\geq V(x,f^{*},g^{*})\geq V(x,\pi^{1},g^{*}).$$
    For each fixed $\lambda \in  \mathbb{B}(x)$, let
    $\pi^{2}=\{\pi_{n}^{2},n\geq 0\}$ with $\pi_{0}^{2}=\lambda$ and
    $\pi_{n}^{2}=g^{*}, n\geq 1$, then by Lemma~\ref{lem4}, for each
    $x\in X$, we have
    \begin{equation*}
    V(x,f^{*},g^{*}) \leq V(x,f^{*},\pi^{2})=T(f^{*},\lambda)V(x,f^{*},g^{*}),
    \end{equation*}
    which yields
 \begin{equation*}
V(x,f^{*},g^{*})\leq \min _{\lambda \in \mathbb{B}(x)} T(f^{*},\lambda)V(x,f^{*},g^{*})\leq TV(x,f^{*},g^{*}).
\end{equation*}
    Similarly, we can prove
    $$V(x,f^{*},g^{*})\geq TV(x,f^{*},g^{*}). $$
    Combining the last two inequalities, we obtain the desired result.

    ($\Leftarrow$)

     This part holds, which has been proved in part $(a)$.
\end{proof}

\section{Algorithm}\label{sec4}
In this section, we develop an iterative algorithm to approach to
the value function and Nash equilibrium of our two-person zero-sum
stochastic SMG, where numerically solving matrix games is
iteratively utilized at every state in a form of value iteration.
First, we introduce some concepts about matrix games
\citep{Barron2013}.

A two-person zero-sum static game in a matrix form means that there
is a matrix ${A}=(a_{ij})_{m \times l}$ of real numbers so that if P1,
the row player chooses to play row $i$, while P2, the column player
chooses to play column $j$, then the payoff to P1 is $a_{ij}$ and
the payoff to P2 is $-a_{ij}$. Every row and column represents a
pure strategy adopted by P1 and P2, respectively. Both players aim to choose strategies that maximize their individual
payoffs. To guarantee the
optimality, we have to consider mixed strategies, where a player
chooses a row or column according to some probability distributions.

\begin{definition}\label{defn7}
    A mixed strategy is a vector  $X=(x_{1},x_{2},\dots,x_{m})$ for P1,
    and $Y=(y_{1},y_{2},\dots,y_{l})$  for P2, where
    $$x_{i}\ge 0, \sum_{i=1}^{m}x_{i}=1~~~ and~~~ y_{j}\ge 0, \sum_{j=1}^{l}y_{j}=1. $$
    The components $x_{i}$ and $y_{j}$ represent the probabilities that row $i$ will
    be chosen by P1 and column
    $j$ will be chosen by P2, respectively. Denote the set of mixed strategies with
    $k$ components by
    $$S_{k}=\{(z_{1},z_{2},\dots,z_{k})~|~z_{i}\ge 0,\sum_{i=1}^{k}z_{i}=1\},\quad k=1,2,\dots.$$
\end{definition}
\begin{definition}\label{defn8}
    Let $X=(x_{1},x_{2},\dots,x_{m})$ be a mixed strategy for P1, and
    $Y=(y_{1},y_{2},\dots,y_{l})$ be a mixed strategy for P2, then the
    expected payoff to P1 is
    \begin{equation*}
    E(X,Y)=XAY^{T}.
    \end{equation*}
    In a two-person zero-sum game, the expected payoff to P2 is
    $-E(X,Y)$.
\end{definition}
Both players aim to choose strategies that maximize their individual
payoffs. P1 wants to choose a strategy to maximize the payoff in the
matrix, while P2 wants to choose a strategy to minimize the payoff
in the matrix.
\begin{definition}\label{defn9}
    The upper and lower values of the matrix game are defined as
    $$v^{+}=\inf_{Y\in S_{l}} \sup_{X\in S_{m}} E(X,Y)~~~and~~~v^{-}=\sup_{X\in S_{m}} \inf_{Y\in S_{l}}E(X,Y).$$
    If $v^{+}=v^{-}$, then the common value is called the value of the
    game and denoted by $v^{*}$ .

    Moreover, a saddle point in mixed strategies is a pair $(X^*,Y^*)\in S_{m}\times S_{l}$, which satisfies
    $$E(X,Y^*)\leq E(X^*,Y^*)\leq E(X^*,Y),\quad \forall X\in S_{m},Y\in S_{l}.$$
\end{definition}

By Theorem $1.3.4$ in  \cite{Barron2013}, we know that any matrix
game has a unique value as well as at least one saddle point. There
is a method of formulating the matrix game as a linear program as
follows  \citep{Barron2013}:

P1 aims to choose a mixed strategy
$X^*=(x_{1}^*,x_{2}^*,\dots,x_{m}^*)$ to maximize the payoff
\begin{equation}\label{equ11}
\left\{
\begin{array}{lr}
\max~~v &  \\
\mbox{subject to } \\
\sum_{i=1}^{m}a_{ij}x_{i}^*\ge v, \quad j=1,2,\dots, l\\
\sum_{i=1}^{m}x_{i}^*=1 \\
x_{i}^*\ge 0, \quad i=1,2,\dots,m.
\end{array}
\right.
\end{equation}
P2 aims to choose a mixed strategy
$Y^*=(y_{1}^*,y_{2}^*,\dots,y_{l}^*)$ to minimize the payoff
\begin{align}\label{equ12}
\left\{
\begin{array}{lr}
\min~~v &  \\
\mbox{subject to } \\
\sum_{j=1}^{l}a_{ij}y_{j}^*\leq v, \quad i=1,2,\dots,m\\
\sum_{j=1}^{l}y_{j}^*=1 \\
y_{j}^*\ge 0, \quad j=1,2,\dots,l.
\end{array}
\right.
\end{align}

We can use the classic algorithms to solve the two linear programs
(\ref{equ11}) and (\ref{equ12}), such as simplex method or interior
point method. Note that the optimal values of $v$ solved by
(\ref{equ11}) and (\ref{equ12}) are always equal. Therefore, the
optimal strategies of P1 and P2 and the value of the game can be
obtained in a straightforward way.

Next, we utilize the above technique of solving matrix games to
study the computation of two-person zero-sum stochastic SMGs, where a value iteration-type algorithm is developed to
approach to the value function $V^*$ and Nash equilibrium
$(\pi^{1}_{*},\pi^{2}_{*})$. To this end, we need to introduce the
following concept.

\begin{definition}\label{defn10}
    Assume that the SMG has a value function $V^{*}$. Then a  pair of strategies
    $(\pi^{1}_{\varepsilon},\pi^{2}_{\varepsilon})\in \Pi_{1}\times
    \Pi_{2}$ is said to be an $\varepsilon$-Nash equilibrium of the game
    if $$\| V(\cdot, \pi^{1}_{\varepsilon},
    \pi^{2}_{\varepsilon})-V^{*}(\cdot)\|_{\omega}<\varepsilon.$$ Moreover,
    $V_{\varepsilon}(\cdot):=V(\cdot, \pi^{1}_{\varepsilon},
    \pi^{2}_{\varepsilon})$ is called the $\varepsilon$-value function of the
    game.
\end{definition}

Consider the mathematical model of SMG discussed in this paper. In
order to numerically approach to the value function and Nash
equilibrium, we simplify the general state and action spaces as
finite case for convenience. Without loss of generality, we assume
that $A(x):=\{a_{1},a_{2},\dots,a_{m}\}$ and
$B(x):=\{b_{1},b_{2},\dots,b_{l}\}$, for any $x\in X
:=\{x_{0},x_{1},\dots,x_{n-1}\}$. Under
Assumptions~\ref{ass1}-\ref{ass4} mentioned in Section~\ref{sec3},
we obtain the Shapley equation as follows
\begin{align}\label{equ13}
V^*(x)=TV^*(x) &=\min _{g \in \Phi_{2}} \max _{f \in \Phi_{1}} G(V^*,x, f, g)\nonumber \\
&=\min _{g \in \Phi_{2}} \max _{f \in \Phi_{1}}\sum_{i=1}^{m} \sum_{j=1}^{l} G(V^*,x, a_{i}, b_{j})f(x,i)g(x,j)\nonumber\\
&=\max _{f \in \Phi_{1}}\min _{g \in \Phi_{2}} \sum_{i=1}^{m} \sum_{j=1}^{l} G(V^*,x, a_{i}, b_{j})f(x,i)g(x,j),\quad\forall x\in X.
\end{align}

For each fixed $x\in X$ and given function $u\in B_{\omega}(X)$, let
$C(u,x)$ be an $m\times l$-dimensional matrix with elements defined
as
$$c(u,x)_{ij}:=G(u,x, a_{i}, b_{j}), \quad i=1,2,\dots,m; \ j=1,2,\dots,l,$$
where $G(u,x, a_{i}, b_{j})$ is defined in  (\ref{equ4}). We further define $f(x):=(f(x,1),f(x,2),\dots,f(x,m))$ as an
$m$-dimensional vector and $g(x):=(g(x,1),g(x,2),\dots,g(x,l))$ as
an $l$-dimensional vector, which are all mixed strategies. According
to (\ref{equ13}), we have
\begin{equation}\label{equ14}
V^*(x) = \min _{g \in \Phi_{2}} \max _{f \in
    \Phi_{1}}f(x)C(V^*,x)g(x)^T=\max _{f \in \Phi_{1}}\min _{g \in
    \Phi_{2}}f(x)C(V^*,x)g(x)^T,
\end{equation}
which can be viewed as a matrix game for the value function $V^*$ at
each state $x \in X$.

However, we cannot directly solve (\ref{equ14}) since the value
function $V^*$ is unknown. Below, we develop Algorithm~\ref{algo1}
to iteratively compute a series of matrix games whose values can
asymptotically approach to $V^*(x)$ at each state $x$. From the
lines 11-12 of Algorithm~\ref{algo1}, we can see that at the $n$th
iteration, we can obtain $V_n(x)$ and $(f_{n}(x),g_{n}(x))$ by using
linear programming (\ref{equ11}) and (\ref{equ12}) to solve the game
with matrix $C(V_{n-1},x)$ whose element is
$c(V_{n-1},x)_{ij}:=G(V_{n-1},x, a_{i}, b_{j})$, where
$n=1,2,\dots$, $i=1,2,\dots,m$, $j=1,2,\dots,l$, and $x \in X$. This
iterative procedure of computing a series of $V_n$ is similar to the
classic value iteration algorithm in the MDP theory.  Furthermore, we
give a theorem (Theorem~\ref{thm2}) to prove the convergence of Algorithm~\ref{algo1}.


\begin{algorithm}[!h]
    \caption{Value iteration-type algorithm to solve the two-person zero-sum SMG}\label{algo1}
    {\bfseries Algorithm parameter:} a small threshold $\epsilon>0$ determining accuracy of estimation; model parameters $\theta,\delta$ given by Assumption~\ref{ass1}, $\alpha_{0}$ given by Assumption~\ref{ass2}(a), and $\eta$ given by Assumption~\ref{ass3}, with $\gamma=1-\delta+\delta e^{-\alpha_{0}\theta}$ and $\varepsilon=\frac{\epsilon}{1-\eta\gamma}$; a measurable function $\omega:X\rightarrow [1,\infty)$ given by Assumption~\ref{ass2}(b)\\
    {\bfseries Initialize:} $V(x) \in \mathbb{R}$  for all $x\in X$ arbitrarily\\
    \Repeat{ $\Delta<\epsilon$}
    {
        $\Delta \gets 0$\\
        {\bfseries Loop} \For{each $x\in X$}
        {
            $v \gets V(x)$\\
            \For{ $i=1;i<m;i++$}
            {
                \For{$j=1;i<l;j++$}
                {
                    $c(v,x)_{ij} \gets G(v,x,a_{i},b_{j})$\
                }
            }
            Solving the game with matrix $C(v,x)$\\
            $V(x) \gets \max\limits_{f\in S_{m}}\min\limits_{g\in S_{l}}fC(v,x)g$   \\
            $(f(x),g(x)) \gets \arg\max\limits_{f\in S_{m}}\min\limits_{g\in S_{l}}fC(v,x)g$\\
            $\Delta \gets \max\{\Delta,\frac{|v-V(x)|}{\omega(x)}\}$\\
        }
    }
    {\bfseries Output:}\\
    ~~~~~~~~ $V_{\varepsilon}(x)=  V(x)$ and $(f_{\varepsilon}(x),g_{\varepsilon}(x)) = (f(x),g(x))$\\
\end{algorithm}

\begin{remark}
    For the case where the state and action spaces are both
    countable, we generally choose $\omega(x)=1$ for convenience(see Example~\ref{examp1}). And the line $13$ of Algorithm~\ref{algo1}  is simplified to $\Delta \gets \max\{\Delta,|v-V(x)|\}$.
\end{remark}

\begin{thm}\label{thm2}
    Under Algorithm~\ref{algo1}, for any given $\epsilon>0$ and initial
    value $V_{0}\in \mathbb{R}$, there exists a non-negative integer
    $N_{\epsilon}=\big(1+\lfloor log_{\eta\gamma}
    (\frac{\epsilon}{\|TV_{0}-V_{0}\|_{\omega}}) \rfloor\big)\mathbb
    I_{TV_{0}\neq V_{0}}$ such that
    $\|V_{N_{\epsilon}+1}-V_{N_{\epsilon}}\|_{\omega}<\epsilon$, which
    implies that Algorithm~\ref{algo1} can converge within
    $N_{\epsilon}$ iterations. Moreover, the strategy pair
    $(f_{\varepsilon},g_{\varepsilon})$ output by Algorithm~\ref{algo1}
    is an $\varepsilon$-Nash equilibrium, where
    $\varepsilon=\frac{\epsilon}{1-\eta\gamma}$.
\end{thm}
\begin{proof}
    According to the iterative formula of Algorithm~\ref{algo1}, we have
    $$\|V_{n+1}-V_{n}\|_{\omega}=\|TV_{n}-TV_{n-1}\|_{\omega}\leq \eta\gamma \|V_{n}-V_{n-1}\|_{\omega},\quad\forall n\ge 1,$$
    which by iteration yields
    $$\|V_{n+1}-V_{n}\|_{\omega} \leq (\eta\gamma)^{n}\|TV_{0}-V_{0}\|_{\omega},\quad\forall n\ge 0. $$
    For each given $\epsilon>0$ and initial value $V_{0}\in \mathbb{R}$, if $TV_{0}=V_{0}$, choose $N_{\epsilon}=0$, and we have
    \begin{equation}
    \|V_{N_{\epsilon}+1}-V_{N_{\epsilon}}\|_{\omega}=0<\epsilon,
    \nonumber
    \end{equation}
    otherwise, if $TV_{0}\neq V_{0}$, choose $N_{\epsilon}=1+\lfloor log_{\eta\gamma}
    (\frac{\epsilon}{\|TV_{0}-V_{0}\|_{\omega}}) \rfloor$, and we have
    \begin{equation}
    \|V_{N_{\epsilon}+1}-V_{N_{\epsilon}}\|_{\omega}\leq (\eta\gamma)^{N_{\epsilon}}\|TV_{0}-V_{0}\|_{\omega}<\epsilon.
    \nonumber
    \end{equation}
    Combining the two cases above, choose $N_{\epsilon}=\big(1+\lfloor log_{\eta\gamma}
    (\frac{\epsilon}{\|TV_{0}-V_{0}\|_{\omega}}) \rfloor\big)\mathbb I_{TV_{0}\neq V_{0}}$ and we have $\|V_{N_{\epsilon}+1}-V_{N_{\epsilon}}\|_{\omega}<\epsilon$, which implies that Algorithm~\ref{algo1} can converge within
    $N_{\epsilon}$ iterations.

    Moreover, since $V^{*}$ is the unique solution of the Shapley equation, we have
    \begin{equation}
    \|V_{n}-V^*\|_{\omega} \leq \|V_{n+1}-V^*\|_{\omega}+\|V_{n}-V_{n+1}\|_{\omega}\leq \eta\gamma\|V_{n}-V^*\|_{\omega}+\|V_{n}-V_{n+1}\|_{\omega}
    \nonumber
    \end{equation}
    thus,
    \begin{equation}
    \|V_{n}-V^*\|_{\omega} \leq \frac{\|V_{n}-V_{n+1}\|_{\omega}}{1-\eta\gamma},
    \nonumber
    \end{equation}
    taking $n=N_{\epsilon}$, and we have
    \begin{equation}
    \|V_{N_{\epsilon}}-V^*\|_{\omega} < \frac{\epsilon}{1-\eta\gamma}=\varepsilon,
    \nonumber
    \end{equation}
    which implies that $(f_{\varepsilon},g_{\varepsilon})$ is an $\varepsilon$-Nash equilibrium  by Definition~\ref{defn10}.
\end{proof}

Therefore, with Algorithm~\ref{algo1}, we can iteratively approach
to the value function and Nash equilibrium of our SMG problem
through recursively solving linear programming (\ref{equ11}) and
(\ref{equ12}) at each state $x$. Theorem~\ref{thm2} guarantees the
convergence of Algorithm~\ref{algo1}. We can implement
Algorithm~\ref{algo1} with discretization techniques for computers
to solve practical problems, as illustrated in the next section.

\section{Numerical Experiment}\label{sec5}

In this section, we conduct numerical examples to illustrate our
main results derived in Sections~\ref{sec3}\&\ref{sec4}. First, we
give an example to demonstrate that
Assumptions~\ref{ass1}-\ref{ass4} ensuring the existence of the
value function and Nash equilibrium of SMGs are easy to verify in
practice.

\begin{example}\label{examp1}
    Consider a system with a model of SMG which is defined as follows:

    The state space $X:=\{n:n\in \mathbb{N_{+}}\}$ and the action spaces $A=B:=\{n:n\in \mathbb{N_{+}}\}$ with admissible action sets $A(i)=B(i):=\{n:n\in \mathbb{N_{+}}\}$ for each $i\in X$.
    The semi-Markov kernel is given by:
    $$Q(t,j|i, a, b)=\left\{
    \begin{array}{ll}
    (1-e^{-\beta (i,a,b)t})p(j|i,a,b) & {\text{if}\quad i\in \{1,2\}},\\
    {\frac{t}{\beta(i, a, b)} p(j | i, a, b)} & {\text{if}\quad i\geq3, 0 \leq t \leq \beta(i, a, b)}, \\
    {p(j | i, a, b)} & {\text{otherwise}},
    \end{array}
    \right. $$
    where $\beta(i, a, b)$ is a positive constant and $p(\cdot|i,a,b)$ is a probability distribution.
    The payoff function is denoted by $r(i,a,b)$ which is bounded. Moreover, the discount factor is defined as $\alpha(i,a,b):=e^{-\frac{1}{i+a+b}}$.
\end{example}

Now, we verify that the conditions on the existence of a pair of
optimal stationary strategies described in
Assumptions~\ref{ass1}-\ref{ass4} are satisfied in this example. To
this end, we need the following hypothesis:
\begin{assumption}\label{assE}
There exist positive constants $k_1$ and $k_2$ such that for
each $(a,b)\in A\times B$, we have $0<\beta (i,a,b)<k_{1}$ for each
$i\in\{1,2\}$ and $\beta (i,a,b)>k_{2}$ for each $i\geq 3$.
\end{assumption}
With this hypothesis, we directly have the following result.
\begin{proposition}\label{pro1}
    Suppose that Assumption~\ref{assE} holds, then Example~\ref{examp1} satisfies Assumptions~\ref{ass1}-\ref{ass4}, which means the SMG  has a pair of optimal stationary strategies.
\end{proposition}
\begin{proof}
    Obviously, Assumption~\ref{ass2} holds by choosing
    $\alpha_{0}=\frac{1}{4}$ and $M=\sup\limits_{i,a,b}|r(i,a,b)|$. Since $X$ and $A,B$
    are discrete, Assumption~\ref{ass4} holds. Next we verify
    Assumptions~\ref{ass1}\&\ref{ass3}. According to the semi-Markov
    kernel $Q$, we have
    $$H(t|i, a, b)=\left\{
    \begin{array}{ll}
    1-e^{-\beta (i,a,b)t}& {\text{if}\quad i\in \{1,2\}},\\
    {\frac{t}{\beta(i, a, b)}} & {\text{if}\quad i\geq3, 0 \leq t \leq \beta(i, a, b)}, \\
    {1,} & {\text{otherwise}}.
    \end{array}
    \right. $$

    Let $\delta=0.1$ and
    $\theta=\min\{0.9k_{2},\frac{\ln10}{k_{1}}\}$, we have that

    \noindent if $i\in \{1,2\}$,
    $$H(\theta| x, a, b)=1-e^{-\beta (i,a,b)\theta} \leq 1-e^{-k_{1}\frac{\ln10}{k_{1}}}=1-0.1= 1-\delta,$$

    \noindent if $i\geq 3$,
    $$H(\theta| x, a, b)={\frac{\theta}{\beta(i, a, b)}} \leq \frac{0.9k_{2}}{k_{2}}=0.9= 1-\delta,$$

    \noindent which implies that Assumption~\ref{ass1} holds.

    By Lemma~\ref{lem1}, we derive
    $$\gamma=1-(1-e^{-\frac{1}{4}\theta})\delta=\max\{1-0.1(1-0.1^{\frac{1}{4k_{1}}}),1-0.1(1-e^{0.225k_{2}})\}<1$$

    By choosing $\omega(x)=1$ and $\eta=\frac{1+\gamma}{2\gamma}$, we have
    $\eta>1$ and $0<\eta\gamma<1$. Furthermore, for $\forall (i,a,b)\in K$ and $t\ge 0$, we have
    \begin{equation*}
    \begin{aligned}
    \int_{X} \omega(j) Q(t,dj|i,a,b)&=\sum\limits_{j=1}^{+\infty}\omega(j)H(t|i, a, b)p(j|i,a,b) \\
    &=H(t|i, a, b)\\
    &<\eta\omega(i)H(t|i,a,b),
    \end{aligned}
    \end{equation*}
    which yields (\ref{equ3}).

    Therefore, Assumption~\ref{ass3} is also verified. Hence, the SMG of
    Example~\ref{examp1} has a pair of optimal stationary strategies.
\end{proof}

Next, we give another example about investment problem to
demonstrate the numerical computation of Algorithm~\ref{algo1} to
solve the value function and a pair of optimal stationary strategies
of the game.

\begin{example}\label{examp2}
    Consider an investment problem with three
    states $\{1,2,3\}$, which denotes the benefit, medium and loss
    economy environments, respectively. At each state, the investor will buy some assets while the market-maker will sell. The interest rate depends on the economy environments as well as the number of assets that investor buys and market-maker sells. In state $i\in \{1,2\}$, the investor buys a certain amount of assets from $\{a_{i1},a_{i2}\}$ and the market-maker sells
    from $\{b_{i1},b_{i2}\}$, which leads to
    a payoff $r(i,a,b)$ to the investor and $-r(i,a,b)$ to the market-maker, where $a\in \{a_{i1},a_{i2}\}, b\in \{b_{i1},b_{i2}\}$.
    Then the system moves to a new state $j$ with  probability
    $p(j|i,a,b)$ after staying at state $i$ for a random time which
    follows exponential-distribution with parameter $\beta (i,a,b)$. In
    state $3$, the investor buys a certain amount of assets from
    $\{a_{31},a_{32}\}$ and the market-maker sells from $\{b_{31},b_{32}\}$, which leads to a payoff
    $r(3,a,b)$ to the investor and $-r(3,a,b)$ to the market-maker,
    where $a\in \{a_{31},a_{32}\}, b\in \{b_{31},b_{32}\}$. Then the
    system moves to a new state $j$ with probability $p(j|3,a,b)$ after
    staying at state $3$ for a random time uniformly distributed in
    $[0,\beta(3,a,b)]$ with parameter $\beta (3,a,b)>0$. For this
    system, the decision makers aim to find a pair of optimal
    strategies.
\end{example}

First, we establish a model of SMG for this example as follows.

We set $X=\{1,2,3\}$, $A(i)=\{a_{i1},a_{i2}\}$,
$B(i)=\{b_{i1},b_{i2}\}$ for each $i\in X$ and the semi-Markov
kernel $Q$ is given by:
$$Q(t,j|i, a, b)=\left\{
\begin{array}{ll}
(1-e^{-\beta (i,a,b)t})p(j|i,a,b) & {\text{if}\quad i\in \{1,2\}},\\
{\frac{t}{\beta(i, a, b)} p(j | i, a, b)} & {\text{if}\quad i=3, \ 0 \leq t \leq \beta(i, a, b)}, \\
{p(j | i, a, b)} & {\text {otherwise}},
\end{array}
\right. $$
from which we can obtain
\begin{align}
Q(dt,j|i,a,b)=\left\{
\begin{array}{ll}
p(j|i,a,b)\beta (i,a,b)e^{-\beta(i, a, b)t}dt & {\text{if}\quad i\in \{1,2\}},\\
{\frac{1}{\beta(i, a, b)} p(j | i, a, b)} & {\text{if}\quad i=3, \ 0 \leq t \leq \beta(i, a, b)}, \\
0 & {\text{otherwise}},
\end{array}
\right.
\nonumber
\end{align}
and
$$H(t|i,a,b)=\left\{
\begin{array}{ll}
1-e^{-\beta (i,a,b)t} & {\text{if}\quad i\in \{1,2\}},\\
{\frac{t}{\beta(i, a, b)} }  & {\text{if}\quad i=3, \ 0 \leq t \leq \beta(i, a, b)}, \\
1 & {\text {otherwise}}.
\end{array}
\right. $$

Then by (\ref{equ4}), we have
$$G(u,i, a, b)=\left\{
\begin{array}{ll}
\frac {r(i,a,b)}{\alpha (i,a,b)+\beta (i,a,b)}+\frac {\beta
    (i,a,b)}{\alpha (i,a,b)+\beta (i,a,b)}\sum
\limits_{j=1}^{3}p(j|i,a,b)u(j) & {\text{if}\quad i\in \{1,2\}},\\
\frac {r(3,a,b)}{(\alpha (3,a,b))^{2}\beta (3,a,b)}\Big[\alpha (3,a,b)\beta (3,a,b)-1+e^{-\alpha (3,a,b)\beta (3,a,b)}\Big]\\
+\frac {1-e^{-\alpha (3,a,b)\beta (3,a,b)}}{\alpha (3,a,b)\beta (3,a,b)}\sum \limits_{j=1}^{3}p(j|3,a,b)u(j) & {\text{if}\quad i=3}.
\end{array}
\right. $$

To take numerical calculation for this example, we assume that the
values of model parameters are shown in Table~\ref{tab1}.
\renewcommand{\arraystretch}{1.2}
\begin{table}[!htbp]
    \newcommand{\tabincell}[2]{\begin{tabular}{@{}#1@{}}#2\end{tabular}}
    \centering
    \fontsize{4}{8}\selectfont
    \caption{The values of model parameters}\label{tab1}
    \label{tab:performance_comparison}
    \small
    \begin{tabular}{|p{18mm}<{\centering}|p{8mm}<{\centering}|p{8mm}<{\centering}|p{8mm}<{\centering}|p{8mm}<{\centering}|p{8mm}<{\centering}|p{8mm}<{\centering}|p{8mm}<{\centering}|p{8mm}<{\centering}|p{8mm}<{\centering}|p{8mm}<{\centering}|p{8mm}<{\centering}|p{8mm}<{\centering}|}
        \hline
        state&
        \multicolumn{4}{c|}{1}&\multicolumn{4}{c|}{2}&\multicolumn{4}{c|}{3}\\
        \hline
        action & \tabincell{c}{$(a_{11},$ \\ $b_{11})$}& \tabincell{c}{$(a_{11},$ \\ $b_{12})$} & \tabincell{c}{$(a_{12},$ \\ $b_{11})$}&\tabincell{c}{$(a_{12},$ \\ $b_{12})$}&\tabincell{c}{$(a_{21},$ \\ $b_{21})$}&\tabincell{c}{$(a_{21},$ \\ $b_{22})$}&\tabincell{c}{$(a_{22},$ \\ $b_{21})$}&\tabincell{c}{$(a_{22},$ \\ $b_{22})$}&\tabincell{c}{$(a_{31},$ \\ $b_{31})$}&\tabincell{c}{$(a_{31},$ \\ $b_{32})$}&\tabincell{c}{$(a_{32},$ \\ $b_{31})$}&\tabincell{c}{$(a_{32},$ \\ $b_{32})$}\\
        \hline
        $\alpha(x,a,b)$&0.98&0.96&0.92&0.9&0.78&0.76&0.73&0.7&0.86&0.84&0.89&0.82\\ \hline
        $r(x,a,b)$&40&24&18&33&12&8&10&17&3&5&2&6\\ \hline
        $\beta(x,a,b)$&20&30&11&13&7&8&6.5&4&0.34&0.44&0.55&0.15\\ \hline
        $p(1|x,a,b)$&0&0&0&0&0.46&0.48&0.39&0.3&0.45&0.24&0.43&0.4\\ \hline
        $p(2|x,a,b)$&0.5&0.43&0.32&0.62&0&0&0&0&0.55&0.76&0.57&0.6\\ \hline
        $p(3|x,a,b)$&0.5&0.57&0.68&0.38&0.54&0.52&0.61&0.7&0&0&0&0\\ \hline
    \end{tabular}
\end{table}

Under these data, we can verify that
Assumptions~\ref{ass1}-\ref{ass4} hold by using Proposition~\ref{pro1}. Thus, the existence of the
value function  and Nash equilibrium of the SMG are ensured by
Theorem~\ref{thm1}. Moreover, by Assumption~\ref{assE} and proposition~\ref{pro1}, we can choose $k_{1}=100,k_{2}=0.1,\alpha_{0}=0.25,\delta=0.1$, from which we obtain $\theta=\min\{0.9k_{2},\frac{\ln10}{k_{1}}\}=0.023,\gamma=1-\delta+\delta e^{-\alpha_{0}\theta}=0.9994,\eta\gamma=\frac{1+\gamma}{2}=0.9997$.
Next, we use Algorithm~\ref{algo1} to find the
value function and a pair of optimal stationary strategies of the
game. The detailed steps are listed as follows.

\text {Step} 1: Initialization.

Let $n=0$, and $V_{0}(1)=V_{0}(2)=V_{0}(3)=1$; set a small threshold
$\epsilon:=10^{-4}$, and we have $\varepsilon=\frac{\epsilon}{1-\eta\gamma}=0.33$.

\text {Step} 2: Iteration.

For $n \geq 0$, $(a,b)\in A(i) \times B(i)$, we have
\begin{equation*}
u_{n}(i, a, b)=\frac {r(i,a,b)}{\alpha (i,a,b)+\beta (i,a,b)}+\frac
{\beta (i,a,b)}{\alpha (i,a,b)+\beta (i,a,b)}\sum
\limits_{j=1}^{3}p(j|i,a,b)V_{n}(j),\quad i=1,2, \nonumber
\end{equation*}
\begin{equation*}
\begin{aligned}
u_{n}(3, a, b)&=\frac {r(3,a,b)}{(\alpha (3,a,b))^{2}\beta (3,a,b)}\Big[\alpha (3,a,b)\beta (3,a,b)-1+e^{-\alpha (3,a,b)\beta (3,a,b)}\Big]\\
&+\frac {1-e^{-\alpha (3,a,b)\beta (3,a,b)}}{\alpha (3,a,b)\beta (3,a,b)}\sum \limits_{j=1}^{3}p(j|3,a,b)V_{n}(j).
\end{aligned}\nonumber
\end{equation*}
Then, for each state $i\in\{1,2,3\}$, we solve the linear program
\begin{equation}\label{equ15}
\left\{
\begin{array}{lr}
\max\limits _{f\left(i,a_{i1}\right), f\left(i,a_{i2}\right),v}~v &  \\
\mbox{subject to } \\
v \leq u_{n}\left(i, a_{i1}, b_{i1}\right) f\left(i,a_{i1}\right)+u_{n}\left(i, a_{i2}, b_{i1}\right) f\left(i,a_{i2}\right) \\
v \leq u_{n}\left(i,a_{i1}, b_{i2}\right) f\left(i,a_{i1}\right)+u_{n}\left(i, a_{i2}, b_{i2}\right) f\left(i,a_{i2}\right) \\
f\left(i,a_{i1}\right)+f\left(i,a_{i2}\right)=1 \\
f\left(i,a_{i1}\right) \geq 0, f\left(i,a_{i2}\right) \geq 0 ,
\end{array}
\right.
\end{equation}
with the solution denoted by $\pi_{n}^{1}(\cdot | i)$ where
$\pi_{n}^{1}(a_{i1}| i)=f(i,a_{i1})$, $\pi_{n}^{1}(a_{i2}|
i)=f(i,a_{i2})$.

Also we solve the dual program of (\ref{equ15})
\begin{equation*}
\left\{
\begin{array}{lr}
\min\limits _{g\left(i,b_{i1}\right), g\left(i,b_{i2}\right),z}~~z \\
\mbox{subject to } \\
z \geq u_{n}\left(i, a_{i1}, b_{i1}\right) g\left(i,b_{i1}\right)+u_{n}\left(i, a_{i1}, b_{i2}\right) g\left(i,b_{i2}\right) \\
z \geq u_{n}\left(i,a_{i2}, b_{i1}\right) g\left(i,b_{i1}\right)+u_{n}\left(i, a_{i2}, b_{i2}\right) g\left(i,b_{i2}\right) \\
g\left(i,b_{i1}\right)+g\left(i,b_{i2}\right)=1 \\
g\left(i,b_{i1}\right) \geq 0, g\left(i,b_{i2}\right) \geq 0 ,
\end{array}
\right.
\end{equation*}
with the solution denoted by $\pi_{n}^{2}(\cdot | i)$ where
$\pi_{n}^{2}(b_{i1}| i)=g(i,b_{i1})$, $\pi_{n}^{2}(b_{i2}|
i)=g(i,b_{i2})$. We set
$$V_{n+1}(i)=\sum\limits_{a\in A(i),b\in B(i)}u_{n}(i, a, b)\pi_{n}^{1}(a| i)\pi_{n}^{2}(b| i).$$

\text {Step} 3: Termination judgement.

If $\max\limits_{i=1,2,3}|V_{n+1}(i)-V_{n}(i)|<\epsilon$, then
the iteration stops, $V_{n}$ is the $\varepsilon$-value function and
$(\pi_{n}^{1}(\cdot | i),\pi_{n}^{2}(\cdot | i))$ is
$\varepsilon$-Nash equilibrium of the SMG; Otherwise, set $n = n+1$
and go to Step~2.

We use Matlab to implement the iteration algorithm for this example.
It takes about $10$ seconds to stop at the $93$rd iteration. The
curves of the error of two successive iterations, the value
function, and the strategy pair of players with respect to the
iteration times are illustrated by Figures~\ref{fig1}-\ref{fig3}.

\begin{figure}[H]
    \begin{minipage}[t]{0.5\textwidth}
        \centering
        \includegraphics[scale=0.5]{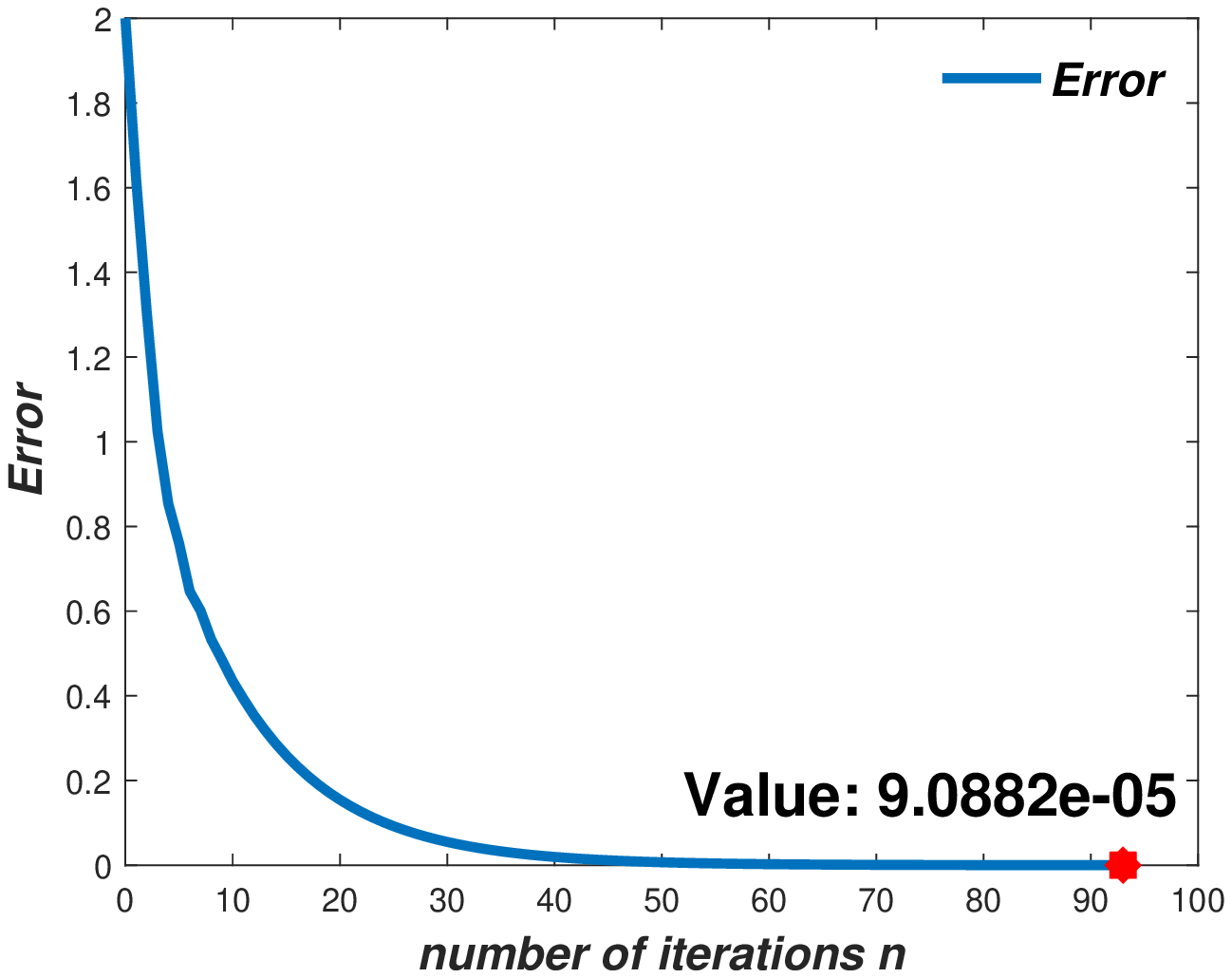}
        \caption{The error}\label{fig1}
    \end{minipage}%
    \hfill
    \begin{minipage}[t]{0.5\textwidth}
        \centering
        \includegraphics[scale=0.5]{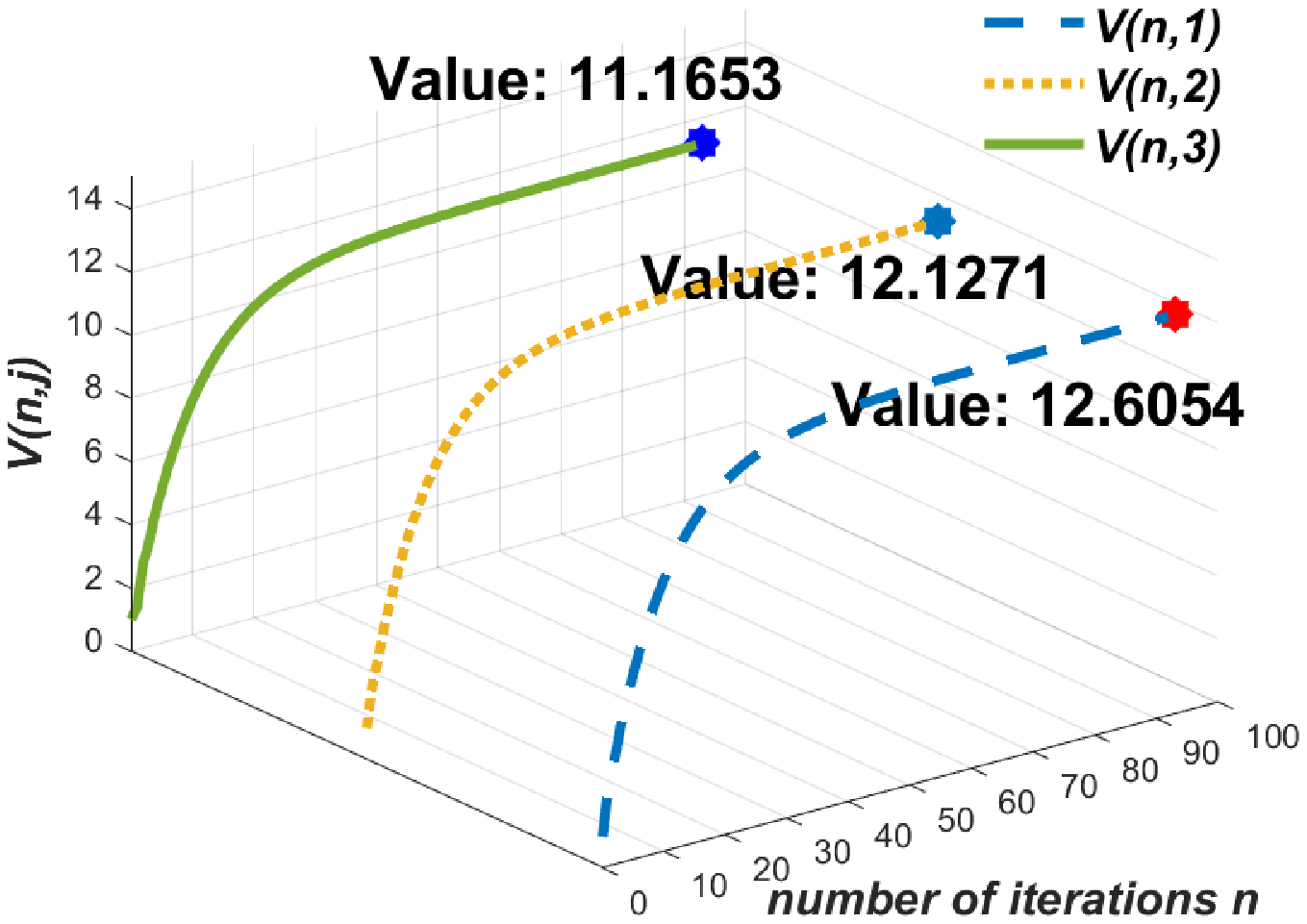}
        \caption{The value function of the game $V$}\label{fig2}
    \end{minipage}
\end{figure}
\begin{figure}[!htbp]
    \centering
    \subfigure{}{
        \begin{minipage}{5.2cm}
            \centering
            \includegraphics[scale=0.4]{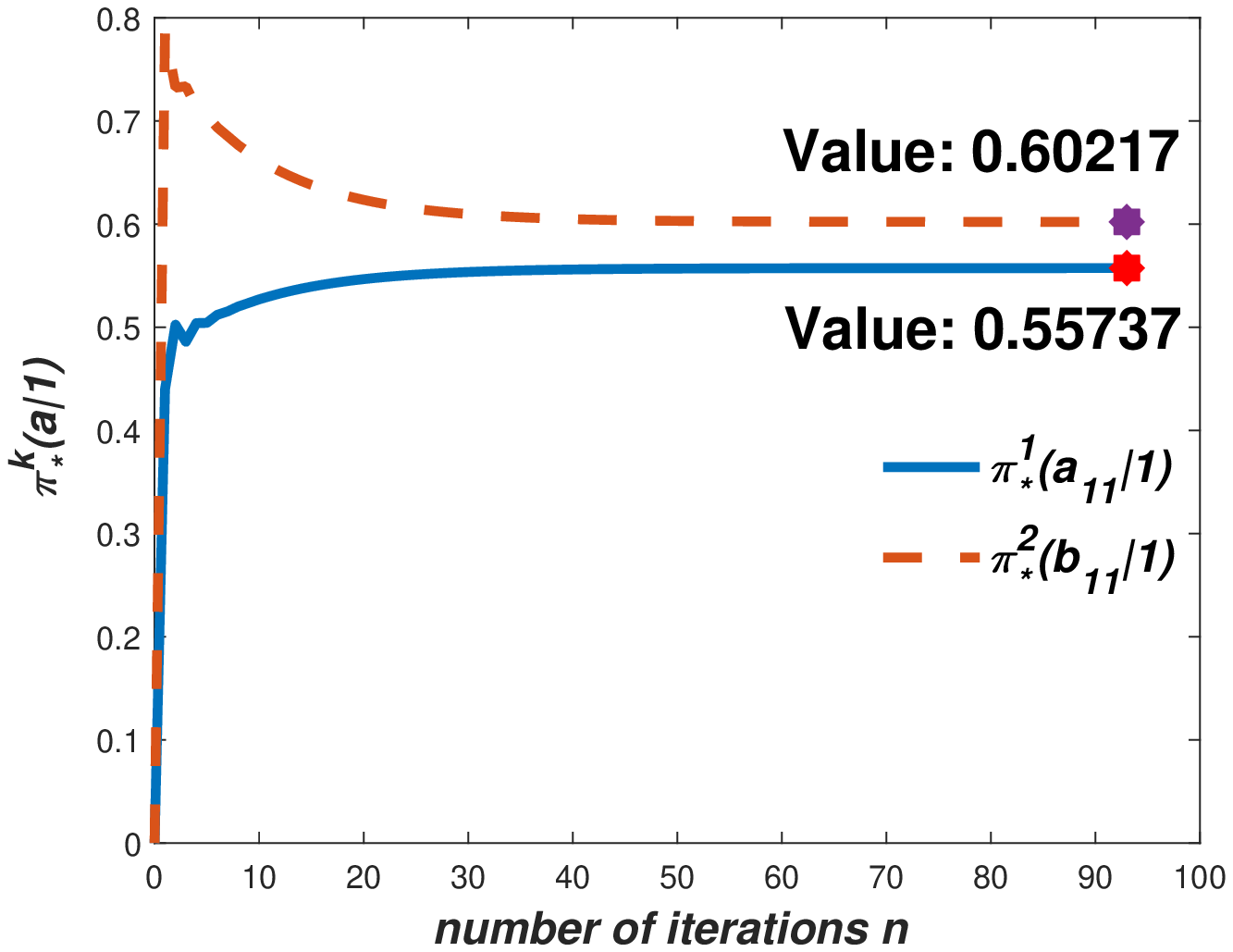}
        \end{minipage}
    }
    \subfigure{}{
        \begin{minipage}{5.2cm}
            \centering
            \includegraphics[scale=0.4]{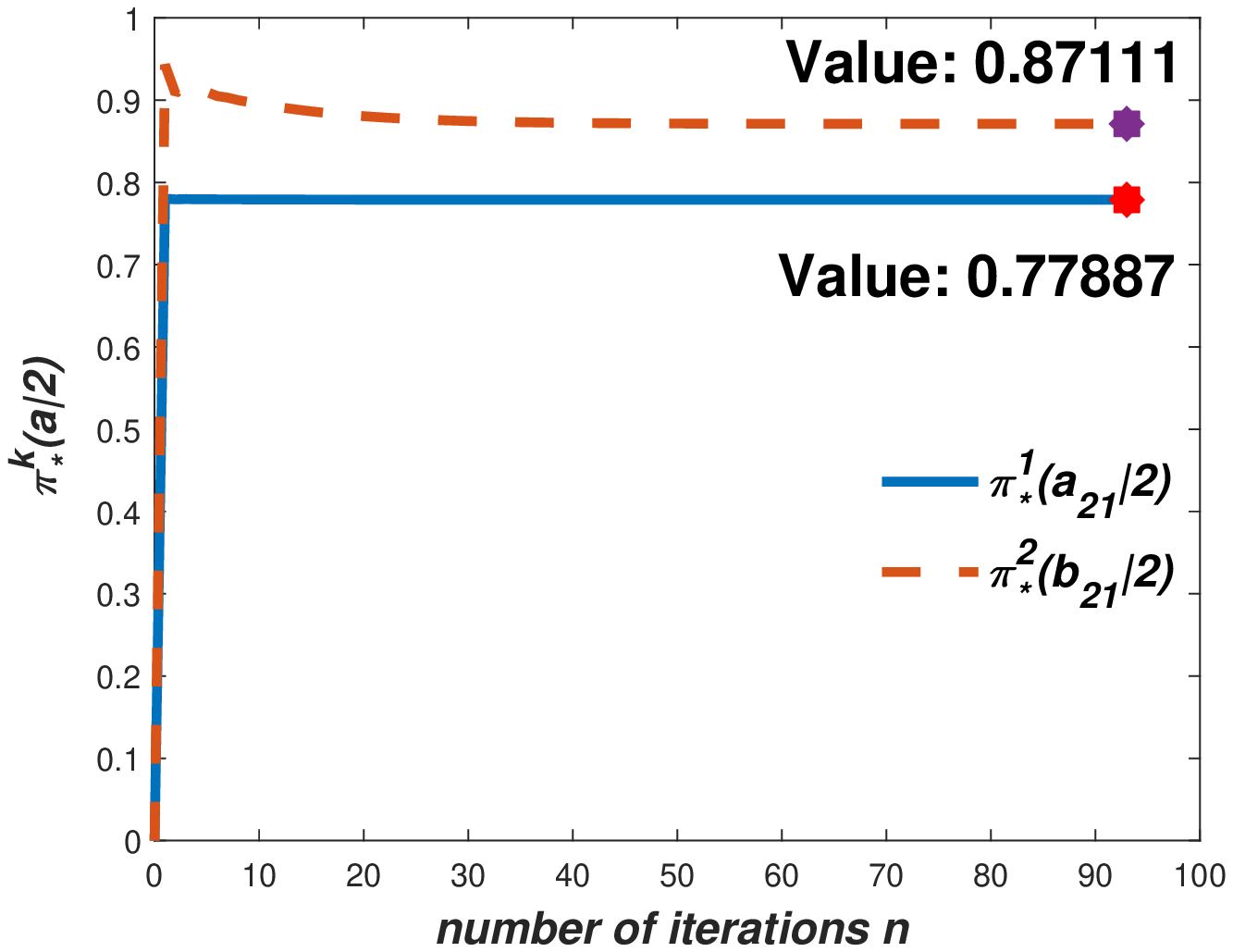}
        \end{minipage}
    }
    \subfigure{}{
        \begin{minipage}{5.2cm}
            \centering
            \includegraphics[scale=0.4]{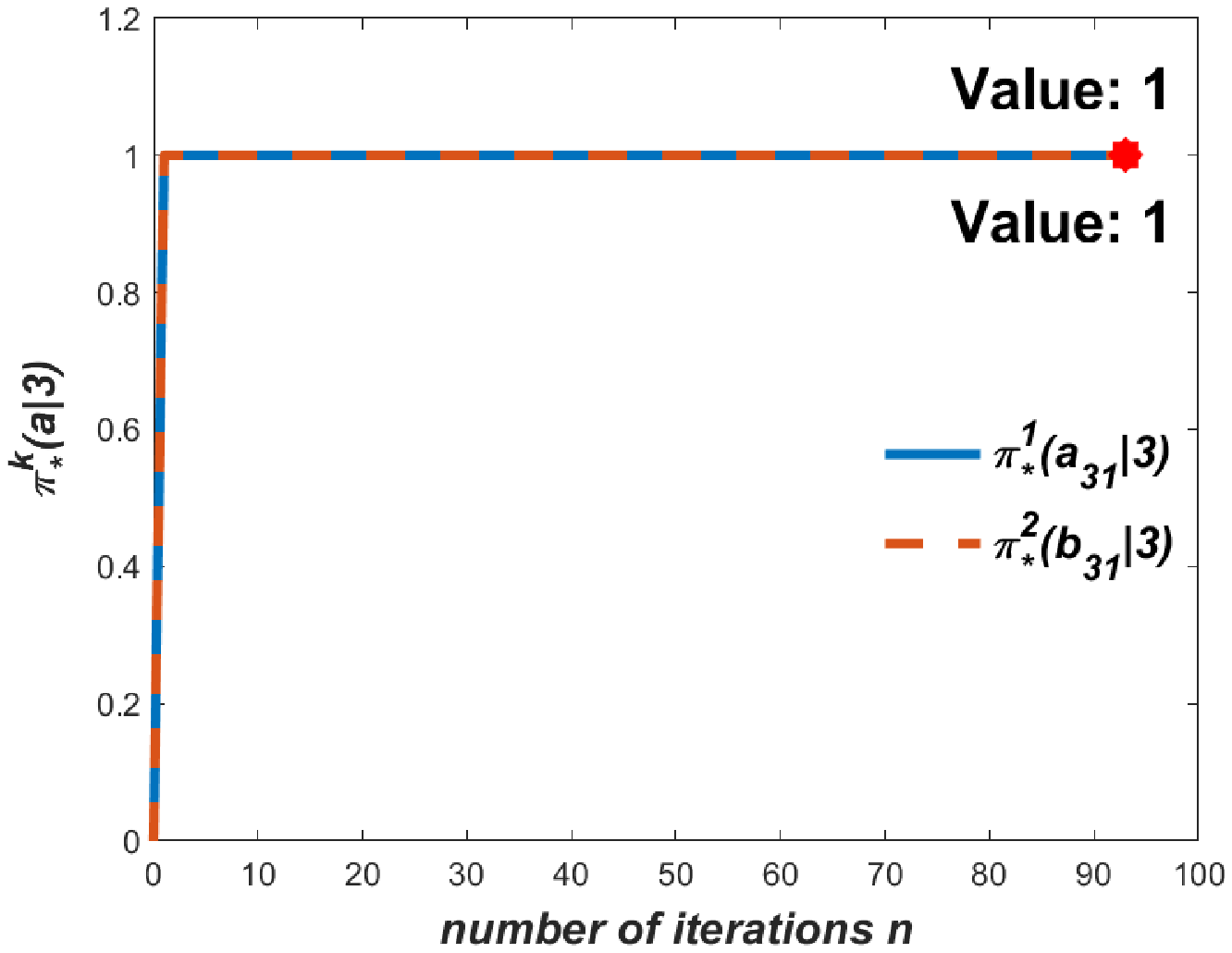}
        \end{minipage}
    }
    \caption{The optimal strategy pair $(\pi_{*}^1,\pi_{*}^2)$}
    \label{fig3}
\end{figure}

Based on the experimental results, we have the following
observations:

$1.$ When the state is benefit, the investor should take action $a_{11}$ with probability $0.60217$ and $a_{12}$ with probability $0.39783$, while the market-maker should take action $b_{11}$ with probability $0.55737$ and $b_{12}$ with probability $0.44263$;

$2.$ When the state is medium, the investor should take action $a_{21}$ with probability $0.87111$ and $a_{22}$ with probability $0.12889$, while the market-maker should take action $b_{21}$ with probability $0.77887$ and $b_{22}$ with probability $0.22113$;

$3.$ When the state is loss, the investor should always take action $a_{31}$ while the market-maker should always take action $b_{31}$;

$4.$ If both investor and market-maker use the optimal
strategies, the investor will obtain a profit $12.6054$ at benefit
state, $12.1271$ at medium state and $11.1653$ at loss state, while the
market-maker will lose the same amount, respectively.
\begin{remark}
    In this example, we choose a uniformly distributed sojourn time at
    state $3$ to show that arbitrary distributions are permitted for the
    sojourn time of semi-Markov processes. Other distributions can also
    be chosen for the sojourn time according to practical situations.
    Moreover, if all the sojourn times are exponentially distributed,
    the semi-Markov games degenerate into discrete-time Markov games.
\end{remark}

\section{Conclusion}\label{sec6}
In this paper, we concentrate on the two-person zero-sum SMG with expected
discounted payoff criterion in which the discount factors are
state-action-dependent. We first construct the SMG model with a fairly
general definition setting. Then we impose suitable conditions on
the model parameters, under which we establish the Shapley equation
whose unique solution is the value function and prove the existence
of a pair of optimal stationary strategies of the game. While the
state and action spaces are finite, a value iteration-type algorithm
for approaching to the value function and Nash Equilibrium is
developed. Finally, we apply our results to an investment problem, which demonstrates that our algorithm performs well.

One of the future research topics is to deal with the nonzero-sum
case of this game model. We wish to find sufficient conditions under
which we use the similar arguments to establish the Shapley equation
and prove the existence of a pair of optimal stationary strategies
for such game. In addition to the value iteration algorithm, the
policy iteration algorithm is also widely used to solve MDPs.
Therefore, it is also promising to develop a policy iteration-type
algorithm to solve the two-person zero-sum SMGs. Moreover,
considering the limitations of computing resources, the dynamic
programming algorithm is difficult to implement in reality when the
scale of the game becomes huge. Another future research topic is to
develop data-driven learning algorithms to approximately solve the
game problems, such as the combination with multi-agent
reinforcement learning approaches.

\section*{Acknowledgements}
This work was supported in part by the National Natural Science
Foundation of China (11931018, 61573206).

\end{document}